\documentclass[sigconf,nonacm]{acmart}
\usepackage{float}
\usepackage{booktabs} 
\usepackage[linesnumbered, ruled, vlined]{algorithm2e}
\usepackage{amsmath}
\usepackage{caption}
\usepackage{subcaption}
\usepackage{graphicx}
\usepackage{tcolorbox}
\usepackage{multicol}
\setcopyright{acmcopyright}
\acmDOI{10.1145/3605098.3636052}
\acmISBN{979-8-4007-0243-3/24/04}
\acmConference[CoDS'25]{13th International Conference on Data Science}{17 – 20 December, 2025}{Pune, India}
\acmYear{2024}
\copyrightyear{2024}
\acmArticle{4}
\acmPrice{15.00}

\begin{document}
\title{Multi-product Influence Maximization in Billboard Advertisement}

\author{Dildar Ali}
\affiliation{%
 \institution{Department of Computer Science and Engineering, Indian Institute of Technology Jammu}
 \city{Jammu}
 \state{Jammu \& Kashmir}
 \country{India}
 \postcode{181221}}
 \email{2021rcs2009@iitjammu.ac.in}

 \author{Rajibul Islam}
\affiliation{%
 \institution{Department of Computer Science and Engineering, Gandhi Institute for Technological Advancement}
 \city{Bhubaneshwar}
 \state{Odisha}
 \country{India}
 \postcode{752054}}
 \email{rajibulislam1604@gmail.com}

\author{Suman Banerjee}
\affiliation{%
 \institution{Department of Computer Science and Engineering, Indian Institute of Technology Jammu}
 \streetaddress{1 Th{\o}rv{\"a}ld Circle}
 \city{Jammu}
  \state{Jammu \& Kashmir}
 \country{India}}
\email{suman.banerjee@iitjammu.ac.in}

\begin{abstract}
Billboard Advertisement has emerged as an effective out-of-home advertisement technique where the goal is to select a limited number of slots and play advertisement content over there with the hope that this will be observed by many people, and effectively, a significant number of them will be influenced towards the brand. Given a trajectory and a billboard database and a positive integer $k$, how can we select $k$ highly influential slots to maximize influence? In this paper, we study a variant of this problem where a commercial house wants to make a promotion of multiple products, and there is an influence demand for each product. We have studied two variants of the problem. In the first variant, our goal is to select $k$ slots such that the respective influence demand of each product is satisfied. In the other variant of the problem, we are given with $\ell$ integers $k_1,k_2, \ldots, k_{\ell}$, the goal here is to search for $\ell$ many set of slots $S_1, S_2, \ldots, S_{\ell}$ such that for all $i \in [\ell]$, $|S_{i}| \leq k_i$ and for all $i \neq j$, $S_i \cap S_j=\emptyset$ and the influence demand of each of the products gets satisfied. We model the first variant of the problem as a multi-submodular cover problem and the second variant as its generalization. For solving the first variant, we adopt the bi-criteria approximation algorithm, and for the other variant, we propose a sampling-based approximation algorithm. Extensive experiments with real-world trajectory and billboard datasets highlight the effectiveness and efficiency of the proposed solution approach.
\end{abstract}

\ccsdesc[500]{Information Systems~Computational Advertising}
\keywords{Billboard Advertisement, Influence Provider, Advertiser, Regret Minimization}
\maketitle

\section{Introduction}
Creating and maximizing influence is an effective marketing strategy, and most commercial houses spend $7-10\%$ of total revenue on advertisements. Among many, one of the effective approaches is billboard advertisement because of its guaranteed return on investment, ease of use, etc. In this advertising technique, a limited number of influential billboard slots are used in the hope that if some attractive advertisement content is displayed, it will create an influence among people who observe the advertisement content. The key problem that arises in this context is how to select a limited number of influential slots so that the influence is maximized. This problem has been referred to as the Influential Billboard Slot Selection Problem. There exist several studies on this problem. 

\paragraph{\textbf{Our Observation.}}
Existing studies in billboard advertising scenarios share a common objective: (a) helping advertisers achieve maximum influence under budget constraints in a single or multi-advertiser setting \cite{zhang2020towards,ali2024effective,kempe2003maximizing}, and (b) minimizing the regret of an influence provider by effective allocation of slots to advertisers \cite{ali2023efficient,ali2024regret,ali2024minimizing,zhang2021minimizing}. It is standard practice that multiple advertisers submit a campaign proposal to the influence provider, specifying the influence demand and the corresponding payment. A more challenging, yet unexplored, problem is how to allocate slots across multiple advertisers having multiple products under a given budget to maximize total influence from the influence provider's perspective.
\paragraph{\textbf{Motivation.}}
Advertisers promote multiple heterogeneous products, each targeting a different customer base. Billboard advertisements typically optimize slot selection to maximize the individual influence of the advertiser or allocate slots to minimize the influence provider’s regret. However, there is a practical need for commercial houses to satisfy distinct influence demands across multiple products simultaneously. This requires a new formulation to allocate billboard slots within a unified budget while maximizing overall influence across all products.

\paragraph{\textbf{Our Problem.}}
We introduce the Multi-Product Influential Billboard Slot Selection Problem in two variants: (1). Common Slot Selection, in which we select a fixed set of slots (within budget) that jointly satisfy the influence demands of all products. (2). Disjoint Slot Selection, where we select disjoint sets of slots, one for each product, ensuring that the individual influence demand of every product is met and the total cost is minimized.

\paragraph{\textbf{Hardness.}}
Both problem variants are computationally hard, generalizing the NP-hard Multi-Submodular Cover Problem \cite{chekuri2022algorithms}. The disjoint variant further extends this to the Generalized Multi-Submodular Cover Problem by adding non-overlapping slot constraints and separate budgets per product, increasing complexity.

\paragraph{\textbf{Our Solutions.}}
We proposed two different solution approaches to solve the slot selection problem. For the Common Slot Selection Problem, we develop a continuous greedy-based bi-criteria approximation algorithm that leverages the submodular structure of the influence function to yield a provable approximation guarantee. For the Disjoint Slot Selection Problem, we design a sampling-based randomized algorithm. This algorithm randomly samples permutations of slot-product assignments and selects the one with minimal cost while satisfying all influence constraints. We perform a sample complexity analysis using Hoeffding’s inequality \cite{hoeffding1963probability}, establishing bounds on the number of samples needed to achieve near-optimal performance with high probability.


\paragraph{\textbf{General Applicability.}}
The proposed framework and algorithms are applicable to any multi-product, out-of-home advertisement setting. This includes billboard networks, public displays, transit advertising, and digital signage, where multiple product campaigns must be scheduled concurrently under budget constraints.

\paragraph{\textbf{Relevant Studies}} 
Existing studies on computational problems in the context of billboard advertisement can be classified into two categories: Influential Slot Selection and Regret Minimization. In the influential slot selection problem, we are given the trajectory information (location of a group of people across different timestamps), billboard slot information (location, duration, and cost), and a limited budget. The task here is to select a limited number of influential billboard slots to maximize the influence.  Existing literature includes a pruned submodularity graph-based approach \cite{ali2022influential,Jali2025influential}, a branch and bound-based framework \cite{ali2024influentialz}, greedy-based solutions \cite{ali2025influential}, and a Co-operative Co-Evolutionary approach proposed by Wang et al. \cite{wang2022data}, etc. On the other direction of studies, it is assumed that there exists an influence provider who has access to a number of billboards and a number of advertiser approaches for an influence demand in exchange for some budget. Now, the task of the influence provider is to allocate slots to the advertisers in such a way that minimizes their total regret. There exists some study on this problem, which includes some local search-based solution approaches by Zhang et al. \cite{zhang2021minimizing}, regret minimization under zonal influence constraint by Ali et al. \cite{ali2024toward,ali2024minimizing,ali2024regret}. However, to our knowledge, there is no literature on the slot selection problem that considers the existence of multiple products that advertisers advertise.
\paragraph{\textbf{Our Contributions.}} 
In practice, a commercial house often advertises multiple diverse products, each targeting different audiences. Thus, billboard slot selection must account for multi-product promotion needs. To the best of our knowledge, there does not exist any study that addresses this problem. To bridge this gap, in this paper, we study the Multi-Product Influential Billboard Slot Selection Problem in two variants. In the first variant, it asks to choose a subset of $k$ billboard slots such that the influence demand of each product gets satisfied and the cost of selection gets minimized. Here, all the selected slots are common, and all of them will be used for promoting the products. We refer to the other variant of the problem as the Disjoint Multi-Product Slot Selection Problem. In this problem, we are given a trajectory and billboard database, and a set of $\ell$ integers. This problem asks to choose nonintersecting $\ell$ many subsets of slots bounded by their respective cardinality, such that the total selection cost is minimized and the influence demand of each product gets satisfied. We have developed approximation algorithms for these problems. In particular, we make the following contributions in this paper:
\begin{itemize}
    \item We extend the influential slot selection problem to the setting where multiple products are to be advertised. We study two related problems in this direction.
    \item As both problems are NP-hard, we propose approximation algorithms to solve these problems.
    \item To address the scalability issues, we have developed efficient heuristic solutions.
    \item  A number of experiments have been carried out to show the effectiveness and efficiency of the proposed solution approaches.
\end{itemize}
\paragraph{\textbf{Organization of the Paper}}
The rest of the paper has been organized as follows. Section \ref{Sec:BPD} describes the required preliminary concepts and defines the problem formally. The proposed solution approaches have been described in Section \ref{Sec:Algo}. Section \ref{Sec:EE} describes the experimental evaluation. Finally, Section \ref{Sec:CFD} presents the concluding remarks of our study.

\section{Background and Problem Definition} \label{Sec:BPD}
\subsection{Trajectory and Billboard Database}
A trajectory database contains the location information of a set of people in a city across different time stamps. In the context to our problem, a trajectory database $\mathcal{D}$ contains tuples of the form $(u_i, \texttt{loc},[t_a,t_b], \mathcal{P}(u_i))$ and this signifies that the person $u_i$ is at the location $\texttt{loc}$ for the duration $t_a$ to $t_b$. Here, $\mathcal{P}(u_i)$ denotes the set of products in which user $u_i$ will be interested in. For any person $u_i$, $loc_{[t_x,t_y]}(u_i)$ denotes the locations of $u_i$ during the time interval $[t_x,t_y]$. Let $\mathcal{U}$ denote the set of people whose movement data is contained in $\mathcal{D}$. Let, $T_{1}= \underset{\tau \in \mathcal{D}}{min} \ t_a$ and $T_{2}= \underset{t \in \mathcal{D}}{max} \ t_b$ and we say that the trajectory database $\mathcal{D}$ contains the movement data from time stamp $T_1$ to $T_2$. A billboard database $\mathcal{B}$ contains the information about billboard slots. Typically, this contains the tuples of the following form $(b_{id}, s_{id}, \texttt{loc}, \texttt{duration})$ where $b_{id}$ and $s_{id}$ denote the billboard id and slot id. \texttt{loc} and \texttt{duration} signify the location of the billboard and the duration of the slot.
\subsection{Billboard Advertisement}
Let a set of $m$ billboards $B=\{b_1, b_2, \ldots, b_m\}$ be placed in different locations of a city. For simplicity, we assume that all the billboards can be leased for a multiple of a fixed duration $\Delta$, which is called a `slot' and has been stated in Definition \ref{Def:Slot}.
\begin{definition} [Billboard Slot] \label{Def:Slot}
A billboard slot is defined by a tuple consisting of two entities, the billboard ID and the duration, that is, $(b_{id}, [t,t+\Delta])$. 
\end{definition}
The set of all billboard slots is denoted by $\mathcal{BS}$, i.e., $\mathcal{BS}=\{(b_i, [t,t+\Delta]): i \in \{1,2, \ldots, m\} \text{ and } t \in \{T_1, T_1+\Delta, T_1+ 2 \Delta, \ldots, T_2-\Delta \}\}$. It can be observed that $|\mathcal{BS}|=m \cdot \frac{T}{\Delta}$ where $T=T_{2}-T_{1}$. For simplicity, we assume that $T$ is perfectly divisible by $\Delta$. For any slot $s_j \in \mathcal{BS}$, $b_{s_j}$ denotes the corresponding billboard and $[t^{s}_{s_j}, t^{f}_{s_j}]$ denotes the time interval for the slot $s_j$ where $t^{f}_{s_j}=t^{s}_{s_j} + \Delta$. Now, we state the notion of the influence probability of a billboard slot in Definition \ref{Def:Inf_Prob}.
\begin{definition} [Influence Probability of a Billboard slot]  \label{Def:Inf_Prob}
Given a slot $s_j \in \mathcal{BS}$ and a person $u_i \in \mathcal{U}$, the influence probability of $s_j$ on $u_i$ is denoted by $Pr(s_j, u_i)$ and can be computed using the following conditional equation.
\[
    Pr(s_j,u_i)= 
\begin{cases}
    \frac{size(s_j)}{\underset{s_k \in \mathcal{BS}}{max} \ size(b_{s_k})},& \text{if } loc_{[t^{s}_{s_j}, t^{f}_{s_j}]}(u_i)=loc(s_j) \\
    0,              & \text{otherwise}
\end{cases}
\]
\end{definition}
Now, we define the influence of a subset of billboard slots in \ref{Def:Influence}.
\begin{definition} [Influence of Billboard Slots] \label{Def:Influence}
Given a trajectory database $\mathcal{D}$, and a subset of billboard slots $\mathcal{S} \subseteq \mathcal{BS}$, the influence of $\mathcal{S}$ can be defined as the expected number of trajectories that are influenced, which can be computed using Equation \ref{Eq:Equation}.

\begin{equation} \label{Eq:Equation}
\mathcal{I}^{'}(\mathcal{S})= \underset{u_i \in \mathcal{U}}{\sum} [1- \underset{s_j \in \mathcal{S}}{\prod} (1-Pr(s_j,u_i))]
\end{equation}
\end{definition}
Here, $Pr(b_j,u_i)$ denotes the influence probability of the billboard slot $b_j$ on the people $u_i$, and $\mathcal{I}^{'}(\mathcal{S})$ denotes the influence of the billboard slots of $\mathcal{S}$. $\mathcal{I}^{'}()$ is the influence function which maps each possible subset of billboard slots to its corresponding influence value, i.e., $\mathcal{I}^{'}: 2^{\mathcal{BS}} \longrightarrow \mathbb{R}_{0}^{+}$ where $\mathcal{I}^{'}(\emptyset)=0$. 
\par Now, it is important to observe that every product may not be relevant to every trajectory user. For effective advertisement, it is important to influence the relevant people for a particular product. Consider $\mathcal{P}= \{1,2, \ldots,p\}$ denotes the set of products. From the trajectory database, for every trajectory user, we have the information about the products for which he is relevant. For a given product, $j \in [p]$, we state the notion of Product Specific Influence for a given billboard slot in Definition \ref{Def:Product_Influence}.
\begin{definition}[Product Specific Influence] \label{Def:Product_Influence}
 Given a specific product and a subset of billboard slots $\mathcal{S} \subseteq \mathcal{BS}$, the product-specific influence of $\mathcal{S}$ for the product $j$ is denoted by $\mathcal{I}_{j}(\mathcal{S})$ and can be computed using Equation \ref{Eq:product_influence}.
\begin{equation} \label{Eq:product_influence}
    \mathcal{I}_{j}(\mathcal{S})=  \underset{u_i \in \mathcal{U}_{k}}{\sum} [1- \underset{s_j \in \mathcal{S}}{\prod} (1-Pr(s_j,u_i))]
\end{equation}
Here, $\mathcal{U}_{k}$ denotes the set of users relevant to the $k$-th product and $\mathcal{I}$ is the product specific influence function, i.e., $\mathcal{I}: 2^{\mathcal{BS}} \times \mathcal{P} \longrightarrow \mathbb{R}^{+}_{0}$. 
\end{definition}

\subsection{Submodular Set Functions and Continuous Extensions}
A function $f: 2^{\mathcal{N}} \rightarrow \mathbb{R}$ is submodular if for every $A \subseteq B \subseteq N$, $f(A \cup \{e\}) - f(A) \geq f(B \cup \{e\}) - f(B)$ such that  $e \in \mathcal{N}$ and  $e \notin B$. Such functions arise in many applications. The continuous extension of submodular set functions has played an important role in algorithmic aspects. The idea is to extend a discrete  set function $f: 2^{\mathcal{N}} \rightarrow \mathbb{R}$ to continuous space $[0,1]^{\mathcal{N}}$. Here, we are mainly concerned with multi-linear extension motivated by maximization problems and refer the interested reader to Calinescu et al. \cite{calinescu2007maximizing}, Vondrak \cite{vondrak2007submodularity} for a detailed discussion.
\par The multi-linear extension is a real-valued set function $f: 2^{\mathcal{N}} \rightarrow \mathbb{R}$, denoted by $F$ and defined as follows:
\begin{equation}
F(x) = \underset{\mathcal{S}\subseteq N}{\sum} f(\mathcal{S}) \underset{i \in \mathcal{S}}{\prod} x_{i} \underset{j \notin \mathcal{S}}{\prod} (1-x_{j})
\end{equation}
where $x \in [0,1]^N$, and the random subset $\mathcal{S} \subseteq N$ is drawn by including each element $i \in N$ independently with probability $x_i$, i.e., $\Pr[i \in \mathcal{S}] = x_i \quad \text{independently for each } i$.
\par For any two vectors $p,q \in [0,1]^{N}$, we use $p \lor q$ and $p \land q$ to denote the coordinate-wise maximum and minimum, respectively of $p$ and $q$. We also make use of the notation $w_{e}(t) \leftarrow F(y(t) \lor 1_{e}) -  F(y(t))$, where $1_{e}$ is the characteristic vector of set \{e\}.

\subsection{Problem Definition}
\paragraph{\textbf{Common and Disjoint Multi-Product Slot Selection Problem.}} As mentioned previously, we have studied the Multi-Product Slot Selection Problem in two variants. First, we talk about the variant where a fixed number of slots will be selected for all the products. In this problem we are given with a trajectory and billboard database (which includes the cost function $w: \mathcal{BS} \longrightarrow \mathbb{R}^{+}$), , and  $\ell$ positive integers $k_1, k_2, \ldots, k_{\ell} \in \mathbb{Z}^{+}$, the task is to find out a subset of the slots $\mathcal{S} \subseteq \mathcal{BS}$ such that the total cost gets minimized and for each product $j \in [\ell]$, their influence demand constraint gets satisfied, i.e., for all $j \in [\ell]$, $I_{j}(\mathcal{S}) \geq k_j$. From a computational point of view, this problem can be posed as follows:

\begin{tcolorbox}
\underline{\textsc{Common Multi-Product Slot Selection Problem}} \\
\textbf{Input:} Billboard $\mathcal{B}$ and Trajectory Database $\mathcal{D}$, Cost Function $w: \mathcal{BS} \longrightarrow \mathbb{R}^{+}$, $\ell$ integers $k_1, k_2, \ldots, k_{\ell} \in \mathbb{Z}^{+}$ for influence demand constraint.

\textbf{Problem:} Select a subset of the slots $\mathcal{S} \subseteq \mathcal{BS}$ such that $\underset{v \in \mathcal{S}}{\sum} w(v)$ gets minimized and the constraints $f_{j}(\mathcal{S}) \geq k_j$ are satisfied.
\end{tcolorbox}

In the Disjoint Multi-Product Slot Selection Problem, along with the trajectory and billboard database, and the cost function, we are given with two sets of $\ell$ integers $p_1, p_2, \ldots, p_{\ell} \in \mathbb{Z}^{+}$ and $\sigma_1, \sigma_2, \ldots, \sigma_{\ell} \in \mathbb{Z}^{+}$. The task here is to select $\ell$ many subsets of slots $\mathcal{S}_1, \mathcal{S}_2, \ldots, \mathcal{S}_{\ell} \subseteq \mathcal{BS}$ such that for all $j \in [\ell]$, $|\mathcal{S}_{j}| \leq k_j$, for all $i \neq j$, $\mathcal{S}_{i} \cap \mathcal{S}_{j} = \emptyset$, and for all $j \in [\ell]$, $I_{j}(\mathcal{S}_{j}) \geq \sigma_{j}$. From a computational point of view, this problem can be posed as follows: 
\begin{tcolorbox}
\underline{\textsc{Disjoint Multi-Product Slot Selection Problem}} \\
\textbf{Input:} Billboard $\mathcal{B}$ and Trajectory Database $\mathcal{D}$, Two sets of $\ell$ integers $p_1, p_2, \ldots, p_{\ell} \in \mathbb{Z}^{+}$ and $\sigma_1, \sigma_2, \ldots, \sigma_{\ell} \in \mathbb{Z}^{+}$.

\textbf{Problem:} Select $\ell$ many subsets of slots $\mathcal{S}_{1}, \mathcal{S}_{2}, \ldots, \mathcal{S}_{\ell} \subseteq \mathcal{BS}$ such that:
\begin{itemize}
    \item [1.] For all $j \in [\ell]$, $|\mathcal{S}_{j}| \leq k_j$
    \item [2.] For all $i \neq j$, $\mathcal{S}_{i} \cap \mathcal{S}_{j}=\emptyset$
    \item [3.] For all $j \in [\ell]$, $I_{j}(\mathcal{S}_{j}) \geq \sigma_{j}$
\end{itemize}
\end{tcolorbox}

As our problem is closely related to the Multi-submodular Cover Problem, first, we describe and generalize it.
\paragraph{\textbf{Multi-Submodular Cover Problem.}} In Multi-Submodular Cover Problem, we are given with a ground set $\mathcal{V}=\{v_1, v_2, \ldots, v_n\}$, a cost function $w:\mathcal{V} \longrightarrow \mathbb{R}^{+}$ which maps each ground set element to its corresponding cost, a set of $\ell$ submodular functions $f_1, f_2, \ldots, f_{\ell}$, and a set of $\ell$ integers $k_1, k_2, \ldots, k_{\ell} \in \mathbb{Z}^{+}$. This problem asks to choose a subset $\mathcal{S}$ of ground set elements, i.e., $\mathcal{S} \subseteq \mathcal{V}$, such that the total cost of the chosen elements gets minimized and for any $i \in [\ell]$, $f_{i}(S) \geq k_i$, i.e., the constraints on the functional values of the submodular functions get satisfied.  This is an optimization problem which can be posed as follows:
\begin{equation}
    \mathcal{S}^{*} \longleftarrow \underset{\substack{\mathcal{S} \subseteq \mathcal{V} \\ \forall j \in [\ell], f_{j}(\mathcal{S}) \geq k_j} }{argmin} \ \underset{\substack{v \in \mathcal{S} }}{\sum} w(v)
\end{equation}

Here, $\mathcal{S}^{*}$ denotes the optimal solution for this problem. From the computational point of view, this problem can be posed as follows:

\begin{tcolorbox}
\underline{\textsc{Multi-Submodular Cover Problem}} \\
\textbf{Input:} The ground set $X=\{x_1, x_2, \ldots, x_n\}$, the cost function $w: X \longrightarrow \mathbb{R}^{+}$, $\ell$ many submodular functions $f_1, f_2, \ldots, f_{\ell}$ defined on the ground set $X$, $\ell$ many positive integers $k_1, k_2, \ldots, k_{\ell} \in \mathbb{Z}^{+}$.

\textbf{Problem:} Select a subset of the ground set elements $\mathcal{S} \subseteq X$ which minimizes $\underset{x \in \mathcal{S}} {\sum} w(x)$ but satisfies $f_{j}(\mathcal{S}) \geq k_j$ for all $j \in [\ell]$.
\end{tcolorbox}
An instance of the Multi-Submodular Cover Problem is said to be $r$-sparse if each ground set element can be involved in a maximum of $ r$ objectives.  This problem has been studied by Chekuri et al. \cite{chekuri2022algorithms}, and they proposed a randomized bi-criteria approximation algorithm which runs in polynomial time. The formal result has been stated in Theorem \ref{Th:1}.
\begin{theorem} \label{Th:1}
    There exists a randomized bi-criteria approximation algorithm that for a given instance of the Multi-Submodular Cover Problem it produces a set $\mathcal{S} \subseteq X$ such that (i) For all $j \in [\ell]$, $f_{j}(\mathcal{S}) \geq (1-\frac{1}{e}-\epsilon)k_j$ and (ii) $\mathbb{E}[w(\mathcal{S})] \geq \mathcal{O}(\frac{1}{\epsilon} ln r)$.
\end{theorem}
\begin{tcolorbox}
\underline{\textsc{Generalized Multi-Submodular Cover Problem}} \\
\textbf{Input:} The ground set $X=\{x_1, x_2, \ldots, x_n\}$, the cost function $w: X \longrightarrow \mathbb{R}^{+}$, $\ell$ many submodular functions $f_1, f_2, \ldots, f_{\ell}$ defined on the ground set $X$, two sets of $\ell$ positive real numbers each  $k_1, k_2, \ldots, k_{\ell} \in \mathbb{Z}^{+}$ and $\sigma, \sigma_2, \ldots, \sigma_{\ell} \in \mathbb{R}^{+}$.

\textbf{Problem:} select $\ell$ many subsets of the ground set $\mathcal{S}_{1}, \mathcal{S}_{2}, \ldots, \mathcal{S}_{\ell} \subseteq X$ such that: 
\begin{itemize}
    \item For all $i,j \in [\ell]$ and $i \neq j$, $\mathcal{S}_{1} \cap \mathcal{S}_{2} = \emptyset$
    \item   For all $i \in [\ell]$, $\underset{x \in \mathcal{S}_{i}}{\sum} \ w(x) \leq k_i$
    \item  For all $i \in [\ell]$, $f_{i}(\mathcal{S}_{i}) \geq k_i$
\end{itemize}
\end{tcolorbox}
\paragraph{\textbf{Generalized Multi-Submodular Cover Problem}} In this problem we are given with a ground set of $n$ elements $X=\{x_1, x_2, \ldots, x_n\}$, a cost function $w$ that maps each ground set element to its corresponding cost; i.e., $w: X \longrightarrow \mathbb{R}^{+}$, $\ell$ many submodular functions $f_1, f_2, \ldots, f_{\ell}$ defined on the ground set $X$, and two sets of $\ell$ many real numbers $k_1, k_2, \ldots, k_{\ell} \in \mathbb{R}^{+}$ and $\sigma, \sigma_2, \ldots, \sigma_{\ell} \in \mathbb{R}^{+}$. The task is to select $\ell$ many subsets of the ground set $\mathcal{S}_{1}, \mathcal{S}_{2}, \ldots, \mathcal{S}_{\ell} \subseteq X$ such that: (i) For all $i,j \in [\ell]$ and $i \neq j$, $\mathcal{S}_{1} \cap \mathcal{S}_{2} = \emptyset$, (ii) For all $i \in [\ell]$, $\underset{x \in \mathcal{S}_{i}}{\sum} \ w(x) \leq k_i$, and (iii) For all $i \in [\ell]$, $f_{i}(\mathcal{S}_{i}) \geq k_i$.

\section{Proposed Solution Approaches} \label{Sec:Algo}
In this section, we describe the proposed solution approaches. Section \ref{Sec:Common} and  \ref{Sec:Disjoint} contain the methodologies for the Common and Disjoint Multi-Product Slot Selection Problem, respectively.

\subsection{Common Multi-Product Slot Selection} \label{Sec:Common}
We analyze the continuous greedy algorithm for general monotone, submodular functions. At each step $t$, in Line No. $3$ to $5$, the marginal gain for each $e \in \mathcal{BS}$ is computed. Next, in Line No. $6$, solve the linear program to get a direction vector, and for each slot in $\mathcal{BS}$, each coordinate of $y(t)$ is updated in Line No. $7$ to $9$. Finally, after time $T$, it returns the fractional solution $y(T)$. Parameter $T$ is the stopping criterion of Algorithm \ref{Algo:1} as mentioned in Theorem \ref{Th:CG}. It is important to note that in Line 8, we modify the direction $y$. The definition of $\delta$ implies that $\delta^{-1}$ is at least $n^{5}$ and is divisible by $T^{-1}$. This ensures that after $T \delta^{-1}$ iterations, $t$ will be exactly equal to $T$.
\begin{theorem}\label{Th:CG} \cite{feldman2011unified}
 For any normalized monotone submodular function $f: 2^{\mathcal{N}} \rightarrow \mathbb{R}^{+}$, a solvable polytope $\mathcal{P}\in [0,1]^{\mathcal{N}}$ and stopping time $T > 0$, Algorithm \ref{Algo:1} finds a point $x \in [0,1]^{\mathcal{N}}$, such that $F(x) \geq [1-e^{-T} -\mathcal{O}(1)] \cdot f(OPT)$.
\end{theorem}

\begin{algorithm}[h]
\scriptsize
\SetAlgoLined
\KwData{Trajectory Database $\mathcal{D}$, Billboard Database $\mathcal{B}$, Influence Functions $\mathcal{I}()$, Billboard Slots $\mathcal{BS}$}
\KwResult{ A solution of vector $\mathcal{Y}(T)$}
Initialize $t \leftarrow 0$, $\mathcal{Y}(0) \leftarrow 1_{\emptyset}$, $n \leftarrow |\mathcal{BS}|$, $\delta \leftarrow T(\lceil n^{5}T \rceil)^{-1}$ \;

\While{$t< T$}{
\For{$\text{each}~ e \in \mathcal{BS}$}{
$w_{e}(t) \leftarrow F(y(t) \lor 1_{e}) -  F(y(t)) $
}
$I(t) \leftarrow argmax\{x \cdot w(t) ~|~ x \in \mathcal{P}\}$\;
\For{$\text{each}~ e \in \mathcal{BS}$}{
$y_{e}(t+\delta) \leftarrow y_{e}(t) + \delta I_{e}(t) \cdot (1 - y_{e}(t))$
}
$t \leftarrow t + \delta$
}

\Return $y(T)$\;
\caption{Continuous Greedy Algorithm}
\label{Algo:1}
\end{algorithm}

\par We address the \textit{Common Multi-Product Slot Selection Problem} in Algorithm \ref{Algo:2}, where a single slot set must satisfy the influence demands of multiple products under budget. We propose a bi-criteria approximation algorithm combining continuous greedy and randomized rounding. First, in Line No. $1$, we scale each product’s influence function $\mathcal{I}_j$ such that $\mathcal{I}_j(\mathcal{BS}) = 1$, and set $k_j \leftarrow 1$, and in Line No. $2$, we define the polytope $\mathcal{P}$. In Line No. $3$, using Algorithm \ref{Algo:1}, we solve the multilinear extension over the matroid polytope to obtain a fractional solution $x \in [0,1]^{|\mathcal{BS}|}$. In the rounding step, Line No. $5$ to $7$, sample $\ell = \lceil \log_{1/(1 - \epsilon)}(r) \rceil$ random subsets from $x$ to form $\mathcal{S}$. Next, in Line No. $8$ to $14$ for any product $j$ where $\mathcal{I}_j(S)$ is below $(1 - \frac{1}{e} - 2\epsilon)k_j$, add slots greedily by marginal gain until the bound $(1 - \frac{1}{e})k_j$ is met.

\begin{algorithm}[h!]
\scriptsize
\caption{Bi-criteria Approximation Algorithm for Multi-Product Influence Maximization}
\label{Algo:2}
\KwData{ Trajectory Database $\mathcal{D}$, Billboard Slots $\mathcal{BS}$, Influence Functions $\{\mathcal{I}_j(\cdot)\}_{j=1}^h$, Cost Function $w: \mathcal{BS} \to \mathbb{R}^+$, Sparsity $r$, Approximation Parameter $\epsilon > 0$, Influence Thresholds $\{k_j\}_{j=1}^h$}
\KwResult{ A set of billboard slots $\mathcal{S} \subseteq \mathcal{BS}$ satisfying all influence constraints approximately}
\textbf{Normalize:} For each $j \in [h]$, set $\mathcal{I}_j(\mathcal{BS}) = 1$ and scale $k_j \gets 1$\;
Let $\mathcal{P} \subseteq [0,1]^{|\mathcal{BS}|}$ be the constraint polytope\;
$x \gets$ ContinuousGreedy$(\mathcal{BS}, \mathcal{P}, T=1)$\;
$\alpha \gets 1 - \epsilon$, $\ell \gets \lceil \log_{1/\alpha} r \rceil$, $S \gets \emptyset$\;
\For{$t = 1$ to $\ell$}{
    Sample $\mathcal{S}_t \subseteq \mathcal{BS}$ by including each $i \in \mathcal{BS}$ independently with probability $x(i)$\;
}
$\mathcal{S} \gets \mathcal{S} \cup \mathcal{S}_{1} \cup  \mathcal{S}_{2} \cup \ldots \cup \mathcal{S}_t$\;
\For{$j = 1$ to $h$}{
    \If{$\mathcal{I}_j(\mathcal{S}) < (1 - \frac{1}{e} - 2\epsilon) \cdot k_j$}{
        $M_j \gets \emptyset$\;
        \While{$\mathcal{I}_j(M_j) < (1 - \frac{1}{e}) \cdot k_j$}{
            $u^* \gets \arg\max_{u \in \text{BS} \setminus M_j} \left[ \mathcal{I}_j(M_j \cup \{u\}) - \mathcal{I}_j(M_j) \right]$\;
            $M_j \gets M_j \cup \{u^*\}$\;
        }
        $\mathcal{S} \gets \mathcal{S} \cup M_j$\;
    }
}
\Return $\mathcal{S}$
\end{algorithm}

\begin{theorem}\cite{chekuri2022algorithms}
Let $r$ be the maximum number of functions any slot contributes to (i.e., sparsity), and let $\epsilon > 0$. Then with high probability, the Algorithm \ref{Algo:2} returns a set $S$ such that for all $j \in [\ell]$, $I_j(S) \geq \left(1 - \frac{1}{e} - \epsilon \right) \cdot k_j$, and the expected cost satisfies $\mathbb{E}[w(S)] = O\left( \frac{1}{\epsilon} \log r \right) \cdot \textsc{OPT}$.
\end{theorem}

\paragraph{\textbf{Complexity Analysis.}}
Let $n = |\mathcal{BS}|$ be the number of billboard slots and $\ell$ the number of products. The Continuous Greedy phase runs in $O(n^2 \ell / \epsilon)$ time due to gradient computation over a polytope of size $n$ for each of $\ell$ influence functions. The rounding and union of $\ell = O(\log r / \epsilon)$ random subsets take $O(n \log r / \epsilon)$ time. In the worst case, the repair step examines all $n$ slots per product, taking $O(n \ell)$ time. So the overall time complexity is $O\left( \frac{n^2 \ell}{\epsilon} + n \ell \right)$. The additional space requirement for Algorithm \ref{Algo:2} is $O(n)$ for maintaining the fractional vector and final slot set.

\subsection{Disjoint Multi-Product Slot Selection} \label{Sec:Disjoint}
In this section, we describe the proposed solution approach for the Disjoint Multi-Product Slot Selection Problem. In this problem, we are given a budget and the influence demand for every product. The goal here is to select billboard slots for each product within the budget such that the influence demand for every product is satisfied. Also, a slot can be used for at most one product only. In our proposed methodology, for every possible permutation of products and slots, we perform the following task. As long as the remaining budget of the current advertiser being processed is positive and his influence demand is not satisfied, then in each iteration, we select a slot that makes the maximum possible marginal gain with the existing set of slots that have been chosen for the advertiser. This process terminates once the influence demand is satisfied or the advertiser's budget is exhausted. If the budget of the advertiser is exhausted, then it means that, for the current permutation of advertisers and slots, for this advertiser, the influence demand can not be satisfied. Hence, this does not lead to a feasible solution, and we mark it as an infeasible solution.  This process is continued for all possible permutations of advertisers and slots. Among all the feasible solutions stored in $\mathcal{M}$, we return one with the least cost. Algorithm \ref{Algo:1} describes the whole process in the form of pseudocode. 
\par One point to highlight here is that the number of possible permutations will be of $\mathcal{O}(\ell! \cdot n! )$. By starling's approximation, we can write it as $\mathcal{O}(\ell^{\ell} \cdot n^{n})$. For moderate values of $\ell$ and $n$ (in practice, the value of $n$ could be excessively large), the number of possible permutations could be large. If we consider all of them in the computation, then the execution time requirement will not be practically feasible. To mitigate this problem, in this paper, we sample out independently at random a subset of them. Now, it is an important question what the sample size should be such that the error in the estimation could be bounded with high probability. This leads to the sample complexity analysis, which has been stated subsequently. 

\subsubsection{\textbf{Sample Complexity Analysis.}}
We can use Hoeffding's Inequality to provide a sample bound that guarantees finding a close to optimal solution with high probability \cite{chen2007probability}. Theorem 2 describes the statement.
\begin{theorem}
    Let $X_1, \ X_2, \ \ldots, X_n$ be independent and identically distributed random variables such that for any $i \in [n]$, $a_i \leq X_i \leq b_i$. Let $\bar{X} =\frac{1}{n} \underset{i \in [n]}{\sum} \ X_i$ and $\mu=\mathbb{E}[\bar{X}]$. For any $\epsilon >0$, the following inequality holds 
    \begin{equation}
        Pr[|\bar{X}-\mu| \geq \epsilon] \leq 2 \cdot exp(-\frac{2 n^{2} \epsilon^{2}}{\underset{i \in [n]}{\sum} (b_i - a_i)^{2}}) 
    \end{equation}
\end{theorem}
So for estimation, the upper and lower bound of the solution cost of each feasible solution present in the sample is required. We prove the same in Proposition \ref{Prop:1}. 

\begin{proposition}  \label{Prop:1}
For any feasible solution $\mathcal{S}=(S_1, S_2, \ldots, S_{\ell})$ contained in $\mathcal{M}$ the cost of solution will lie in between $0$ and $\underset{s \in \mathcal{BS}}{\sum} w(s)$.
\end{proposition} 
\begin{proof} For any advertiser $a_i \in \mathcal{A}$, its influence demand $\sigma_{i} \geq 0$ and for any slot $s_j \in \mathcal{BS}$, its cost $w(s_j) \geq 0$. Now, it can be observed that the lower bound holds only when for all the advertisers, their influence demand is all $0$. In that case, every advertiser will be assigned the empty set, i.e., for all $i \in [\ell]$, $S_i=\emptyset$, and hence, the solution cost will be $0$. In any other case, the cost of the solution will be strictly greater than $0$. The upper bound occurs when the allocated slots to the advertisers, i.e., $S_1, S_2, \ldots, S_{\ell}$, is a partition of $\mathcal{BS}$. This means (i) For all $i,j \in [\ell]$, $S_i \cap S_j= \emptyset$, (ii) $\underset{i \in [\ell]}{\bigcup} S_{i}= \mathcal{BS}$. In this case, the cost of the solution will be $\underset{s \in \mathcal{BS}}{\sum} \ w(s)$. In any other case, the cost of the solution will be strictly less than $\underset{s \in \mathcal{BS}}{\sum} \ w(s)$. Hence, for any solution $\mathcal{S}=(S_1, S_2, \ldots, S_{\ell})$, the cost of the solution will be in the range in between $0$ and $\underset{s \in \mathcal{BS}}{\sum} w(s)$.
\end{proof}
\par Now, we prove Theorem \ref{Th:Sample}, which will prove the lower bound on the sample size for which Algorithm \ref{Algo:1} will provide a close to optimal solution with high probability.
\begin{theorem} \label{Th:Sample}
    For any $\epsilon, \delta \in (0,1)$ if the sample size is greater than equals to $\frac{ln(\frac{2}{\delta}) \cdot w(\mathcal{BS})^{2}}{2 \cdot \epsilon^{2} \cdot (W^{\mathcal{A}})^{2}}$ then the probability error in computation will be strictly less than $\epsilon$ with probability at least $(1-\delta)$. 
\end{theorem}
\begin{proof}
    Let, $\mathcal{Y}$ denote the set of sampled solutions. $W^{*}$ and $W^{\mathcal{A}}$ denote the optimal solution and the best solution among the solutions in Sample, respectively. Now, by applying Hoeffding's Inequality and Proposition \ref{Prop:1} we gave the following inequality:
\begin{equation} \label{InEq:Hoeefding}
    Pr[|W^{*}- W^{\mathcal{A}}| \geq \epsilon \cdot W^{*}] \leq  2 \cdot exp(- \frac{2 \cdot \epsilon^{2} \cdot (W^{\mathcal{A}})^{2} \cdot |\mathcal{Y}|}{w(\mathcal{BS})^{2}})
\end{equation}
We want to establish the sample size such that $ Pr[|W^{*}- W^{\mathcal{A}}| < \epsilon \cdot W^{*}] \geq 1- \delta$. By simplifying, we have the following:
\begin{center}
    $2 \cdot exp(- \frac{2 \cdot \epsilon^{2} \cdot (W^{\mathcal{A}})^{2} \cdot |\mathcal{Y}|}{w(\mathcal{BS})^{2}}) \leq \delta$\\
    \vspace{0.2 cm}
    $ln(\frac{2}{\delta}) \leq \frac{2 \cdot \epsilon^{2} \cdot (W^{\mathcal{A}})^{2} \cdot |\mathcal{Y}|}{w(\mathcal{BS})^{2}}$
\end{center}
This leads to the lower bound on the sample size, which is
\begin{center}
    $|\mathcal{Y}| \geq \frac{ln(\frac{2}{\delta}) \cdot w(\mathcal{BS})^{2}}{2 \cdot \epsilon^{2} \cdot (W^{\mathcal{A}})^{2}}$
\end{center} 
\end{proof}

\begin{algorithm}[h]
\scriptsize
\caption{Randomized Algorithm for Multi-Product Influence Maximization}
\label{Algo:3}
\KwData{Billboard Slots $\mathcal{BS}$, Cost Function $w:\mathcal{BS} \longrightarrow \mathbb{R}^{+}$, The influence functions $I_{1}, I_{2}, \ldots, I_{\ell}$, $v_1, v_2, \ldots, v_{\ell} \in \mathbb{R}^{+}$, and $\sigma_1, \sigma_2, \ldots, \sigma_{\ell} \in \mathbb{R}^{+}$}
\KwResult{Non-intersecting sets of slots $S_1, S_2, \ldots, S_{\ell} \subseteq \mathcal{BS}$}

$a \longleftarrow \ell !$, $b \longleftarrow |\mathcal{BS}|!$\;
\For{ $\text{ All } p \in [a] \text{ AND }q \in [b]$}{
$\mathcal{M}[p,q] \longleftarrow \text{ Empty Solution}$\;
$Flag[p,q] \longleftarrow 1$\;
}
\For{ $\text{ All } p \in [a] \text{ AND }q \in [b]$}{
$\mathbb{S} \gets \emptyset$\;
\For{$i=1 \text{ to } \ell$}{
$S_i \longleftarrow \emptyset$\;
}
\For{$i=1 \text{ to } \ell$}{
\While{$v_i> 0 \text{ AND }I(S_{i}) < \sigma_{i}$}{
$s^{*} \longleftarrow \underset{\substack{s \in \mathcal{BS} \setminus \mathbb{S} \\ w(s) \leq v_i}}{argmax} \ I(S_i \cup \{ s\}) \ - \ I(S_{i})$\;
$S_{i} \longleftarrow S_{i} \cup \{s^{*}\}$, $\mathbb{S} \longleftarrow \mathbb{S} \cup \{s^{*}\}$\;
$v_{i} \longleftarrow v_{i} - w(s^{*})$\;
}
\If{$I(S_i) < \sigma_{i}$}
{
$Flag[p,q] \longleftarrow 0$\;
$break$\;
}
}
\eIf{$Flag[p,q] = 0$}
{
$\mathcal{M}[p,q] \longleftarrow \text{ Not a Feasible Solution}$\;
}
{
$\mathcal{M}[p,q] \longleftarrow [\{S_1,S_2, \ldots,S_{\ell} \}, \underset{s \in S_1 \cup S_2 \cup, \ldots,S_{\ell} }{\sum} \ w(s)]$\;
}

}
$(S^{*}_1, S^{*}_2, \ldots, S^{*}_{\ell}) \longleftarrow  \underset{(S_1, S_2, \ldots, S_{\ell}) \in \mathcal{M}}{argmin} \ \ \ \underset{s \in S_1 \cup S_2 \cup, \ldots,S_{\ell}}{\sum} w(s)$\;
$\text{return } (S^{*}_1, S^{*}_2, \ldots, S^{*}_{\ell})$\;
\end{algorithm}

\subsubsection{\textbf{Computational Complexity Analysis.}} Consider $|\mathcal{X}|$ denotes the sample size used to implement Algorithm \ref{Algo:1}. Both the \texttt{for} loops from $2$ to $4$ and from $5$ to $20$ will execute $\mathcal{O}(|\mathcal{X}|)$. Again the for loops from $7$ to $8$ and $9$ to $16$ will execute for $\mathcal{O}(\ell)$ times. Now, it is important to analyze how many times the while loop from $10$ to $16$ will execute. Consider the minimum selection cost among all the billboards is denoted by $w_{min}$, i.e., $w_{min}=\underset{bs_j \in \mathcal{BS}}{min} \ w(bs_j)$. Lat $B=\{k_1, k_2, \ldots, k_{\ell}\}$ are the budgets of the advertisers. Among these, let $B_{max}$ denote the maximum budget among all the advertisers, i.e., $B_{max}=\underset{i \in [\ell]}{max} \ k_i$. It can be observed that the maximum $\frac{B_{max}}{w_{min}}$ number of slots can be allocated to an advertiser. Hence, the while loop can be iterated at max $\mathcal{O}(\frac{B_{max}}{w_{min}})$ times. Within the while loop, the main computation is the marginal influence gain computation of a non-allocated slot, and this requires $\mathcal{O}(n \cdot t)$ time, where $n$ and $t$ denote the number of billboard slots and the tuples in the trajectory database, respectively. Hence, the time requirement for the execution of the while loop will be of $\mathcal{O}(\frac{B_{max}}{w_{min}} \cdot n \cdot t)$. All the remaining statements within the for loop from $5$ to $20$ will take $\mathcal{O}(1)$ time. Hence, the time requirement till Line $20$ will be of $\mathcal{O}(|\mathcal{X}| \cdot \ell \cdot \frac{B_{max}}{w_{min}} \cdot n \cdot t)$. In the worst case, it may so happen that all the sampled solutions are feasible solutions. In Line $21$, we find the best solution among all the feasible solutions in $\mathcal{X}$. It will take $\mathcal{O}(|\mathcal{X}|)$. Hence, the time requirement to execute of Algorithm \ref{Algo:1} is of  $\mathcal{O}(|\mathcal{X}| \cdot \ell \cdot \frac{B_{max}}{w_{min}} \cdot n \cdot t)$. In addition space consumed by Algorithm \ref{Algo:1} is to store the feasible solutions which will be of $\mathcal{O}(|\mathcal{X}| \cdot n)$, to store the sets $S_1, S_2, \ldots, S_{\ell}$ and $\mathbb{S}$ and both of them will take $\mathcal{O}(n)$ space. Hence total extra space consumed by $\mathcal{O}(|\mathcal{X}| \cdot n+n)$ which is equal to $\mathcal{O}(|\mathcal{X}| \cdot n)$. Hence, Theorem \ref{TH:Algo_1_resource} holds.
\begin{theorem} \label{TH:Algo_1_resource}
    The sampling-based approach will take $\mathcal{O}(|\mathcal{X}| \cdot \ell \cdot \frac{B_{max}}{w_{min}} \cdot n \cdot t)$ time and $\mathcal{O}(|\mathcal{X}| \cdot n+n)$ space to execute.
\end{theorem}

\section{Experimental Evaluations} \label{Sec:EE}
This section describes the experimental evaluation of the proposed solution approaches. Initially, we start by describing the datasets.

\subsection{Dataset Descriptions.}
We used two real-world datasets: New York City (NYC)\footnote{\url{https://www.nyc.gov/site/tlc/about/tlc-trip-record-data.page}} and Los Angeles (LA)\footnote{\url{https://github.com/Ibtihal-Alablani}}, both adopted in previous studies \cite{ali2022influential,ali2023influential,ali2024influential,ali2024regret}. The NYC dataset contains 227,428 check-ins (Apr 2012–Feb 2013), while the LA dataset includes 74,170 check-ins across 15 streets. Billboard data for both cities is sourced from LAMAR\footnote{\url{http://www.lamar.com/InventoryBrowser}}, comprising 1,031,040 slots in NYC and 2,135,520 in LA. In Table \ref{Table:Trajectory_Info} $|\mathcal{T}|$, $|\mathcal{U}|$, and $Avg_{dist}$ denote the number of trajectories, the number of unique users, and the average distance between trajectories. Similarly, in Table \ref{Table:Billboard_Info} $|\mathcal{B}|$, $|\mathcal{BS}|$, and $Avg_{dist}$ denote the number of billboards, the number of billboard slots, and the average distance between the billboards placed across the city. The descriptions of the datasets are summarized in Table \ref{Table:Trajectory_Info} and Table \ref{Table:Billboard_Info}.
\begin{table}[h]
\centering
\begin{minipage}{0.48\linewidth}
\centering
\scriptsize
\begin{tabular}{|c|c|c|c|}
\hline
\textbf{Dataset} & $|\mathcal{T}|$ & $|\mathcal{U}|$ & $Avg_{dist}$ \\ \hline
NYC & $227428$ & $1083$ & $3.12~km$ \\ \hline
LA & $74170$ & $2000$ & $0.61~km$ \\ \hline
\end{tabular}
\caption{Trajectory Datasets}
\label{Table:Trajectory_Info}
\end{minipage}
\hfill
\begin{minipage}{0.48\linewidth}
\centering
\scriptsize
\begin{tabular}{|c|c|c|c|}
\hline
\textbf{Dataset} & $|\mathcal{B}|$ & $|\mathcal{BS}|$ & $Avg_{dist}$ \\ \hline
NYC & $716$ & $1031040$  & $15.07~km$ \\ \hline
LA & $1483$ & $2135520$ & $10.26~km$ \\ \hline
\end{tabular}
\caption{Billboard Datasets}
\label{Table:Billboard_Info}
\end{minipage}
\end{table}
\vspace{-0.3 in}
\subsection{Key Parameters.} 
All key parameters used in our experiments are summarized in Table \ref{Table-2}. The default parameter settings are highlighted in bold.
\begin{table}[h!]
\caption{\label{Table-2} Key Parameters}
\vspace{-0.15in}
\centering
\scriptsize  
\begin{tabular}{|c|p{6.2cm}|}
\hline
\textbf{Parameter} & \textbf{Values} \\ \hline
$\alpha$ & $40\%, 60\%, 80\%, 100\%, 120\%$ \\ \hline
$\beta$ & $1\%, 2\%, \textbf{5\%}, 10\%, 20\%$ \\ \hline
$|\mathcal{P}|$ & $100, 50, \textbf{20}, 10, 5$ \\ \hline
$\epsilon$ & $0.01, 0.05, \textbf{0.1}, 0.15, 0.2$ \\ \hline
$\lambda$ & $25\text{m}, 50\text{m}, \textbf{100\text{m}}, 125\text{m}, 150\text{m}$ \\ \hline
\end{tabular}
\end{table}

\paragraph{\textbf{Advertiser}}
Given an advertiser $\mathcal{A}$, with a set of products $\mathcal{P} = \{p_1, p_2,\ldots, p_n \}$, the advertiser can be represented as $(\mathcal{A}, p_{i}, \sigma_{i}, \mathcal{B})$, where $p_{i}$ is the $i^{th}$ product, $\sigma_{i}$ and $\mathcal{B}$ is the corresponding influence demand and budget. 
\paragraph{\textbf{Demand Supply Ratio.}}
The ratio $\alpha = \sigma^{\mathcal{P}} / \sigma^{*}$ represents the proportion of global influence demand $(\sigma^{\mathcal{P}} = \sum_{i=1}^{n} \sigma_i)$ to total influence supply $(\sigma^{*} = \sum_{b \in \mathcal{BS}} \mathcal{I}(b))$. We evaluate five values of $\theta: 40\%, 60\%, 80\%, 100\%, ~and~ 120\%$.
\paragraph{\textbf{Advertiser Individual Demand.}}
The ratio $\beta = \sigma^{\mathcal{P}''} / \sigma^{*}$ denotes the ratio of average individual demand to total influence supply, where $\sigma^{\mathcal{P}''} = \sigma^{\mathcal{P}} / |\mathcal{P}|$ is the average influence demand per product. This value $\beta$ adjusts the demand of individual products.
\paragraph{\textbf{Product Demand.}}
The demand for each product is generated as $\sigma = \lfloor \omega \cdot \sigma^{*} \cdot \beta \rfloor$, where $\omega$ is a random factor between 0.8 and 1.2 to simulate varying product payments.
\paragraph{\textbf{Billboard Cost.}}
The outdoor advertising companies like LAMAR and JCDecuax do not disclose the actual cost of renting a billboard slot. In existing studies \cite{ali2022influential,ali2023influential,ali2024effective,aslay2015viral,banerjee2020budgeted,zhang2019optimizing}, the cost of a billboard slot is proportional to its influence. Following this, we calculated the cost of a billboard slot `bs' using $\lfloor \delta \cdot \frac{\mathcal{
I}(bs)}{10} \rfloor$, where $\delta \in [0.8,1.1]$ and $\mathcal{I}(bs)$ denotes the influence of billboard slot `bs'.
\paragraph{\textbf{Advertiser Budget.}}
Following existing studies \cite{aslay2017revenue,aslay2015viral,zhang2021minimizing,ali2024minimizing,ali2024regret}, we set each product’s payment proportional to its influence demand: $\mathcal{L}_{i} = \lfloor \eta \cdot \sigma_{i} \rfloor$, where $\eta$ is a random factor in [0.9, 1.1] to model payment variations. The total advertiser budget is then $\mathcal{B} = \sum_{i=1}^{n} \mathcal{L}_{i}$.
\paragraph{\textbf{Environment Setup.}}
All Python code is executed on an HP $Z4$ workstation with $64$ GB RAM and a $3.50$ GHz Intel Xeon(R) CPU.

\subsection{Baseline Methodologies.} \label{Sec:Baseline}
\paragraph{\textbf{Random Allocation (RA)}}
In this approach, billboard slots are selected uniformly at random and assigned to products until both the influence demand and budget constraints are met.
\paragraph{\textbf{Top-k Allocation}}
In this approach, billboard slots are first sorted by their influence values. The highest-influence slots are allotted to the products sequentially until the influence requirements and budget constraints are met.

\subsection{\textbf{Goals of Our Experimentation}}
\begin{itemize}
\item \textbf{RQ1}: Varying $\alpha$, $\beta$, how do the satisfied product change?
\item \textbf{RQ2}: Varying $\alpha$, $\beta$, how do the total budget is utilized?
\item \textbf{RQ3}: Varying $\alpha$, $\beta$, how do the computational time change?
\item \textbf{RQ4}: Varying $\epsilon$ and $\lambda$, how do the influence quality change?
\end{itemize}

\subsection{Results and Discussions} \label{Sec:RD} 
\paragraph{\textbf{Varying $\alpha$ and $\beta$ Vs. Influence.}}
From Figure \ref{Fig:NYC_Combined}(a,b) and Figure \ref{Fig:LA_Combined}(a,b), it is clear that with the increase of demand supply ratio $(\alpha)$, within the budget of each product, the influence value increases. When the $\alpha$ value is $100\%$ or $120\%$, the influence demand is higher for each product and requires a larger number of slots to satisfy. In the case of the `RA' approach, in the NYC dataset, it satisfies the influence demand of $100$  products with the lowest cost compared to the `Random' and `Top-K' approaches. In the case of the `BCA' approach, the influence gain is higher compared to the baseline methods like `Random' and `Top-K'; however, it is less than the `RA' approach. This happens because in the `BCA' approach, we find a common set of slots that can satisfy all the products. The difference in influence is $20\%$ to $25\%$ more in the NYC dataset and $8\%$ to $10\%$ more than the baselines in the LA dataset. Now, in case of the LA dataset in Figure \ref{Fig:LA_Combined}, similar observations were found for all the proposed and baseline methods. In the experimental setup, when $\alpha \geq 100\%$ and $\beta \geq 10\%$, the influence demand is greater than the influence supply in those cases, and the influence demand of all the products is not satisfied. In these cases, `BCA' methods perform better than the baselines; however, the `RA' approach fails to provide a feasible solution.
\begin{figure*}[ht]
    \centering
    \setlength{\tabcolsep}{2pt} 
    \renewcommand{\arraystretch}{0.9} 
    \begin{tabular}{cccc}
        \includegraphics[width=0.23\linewidth]{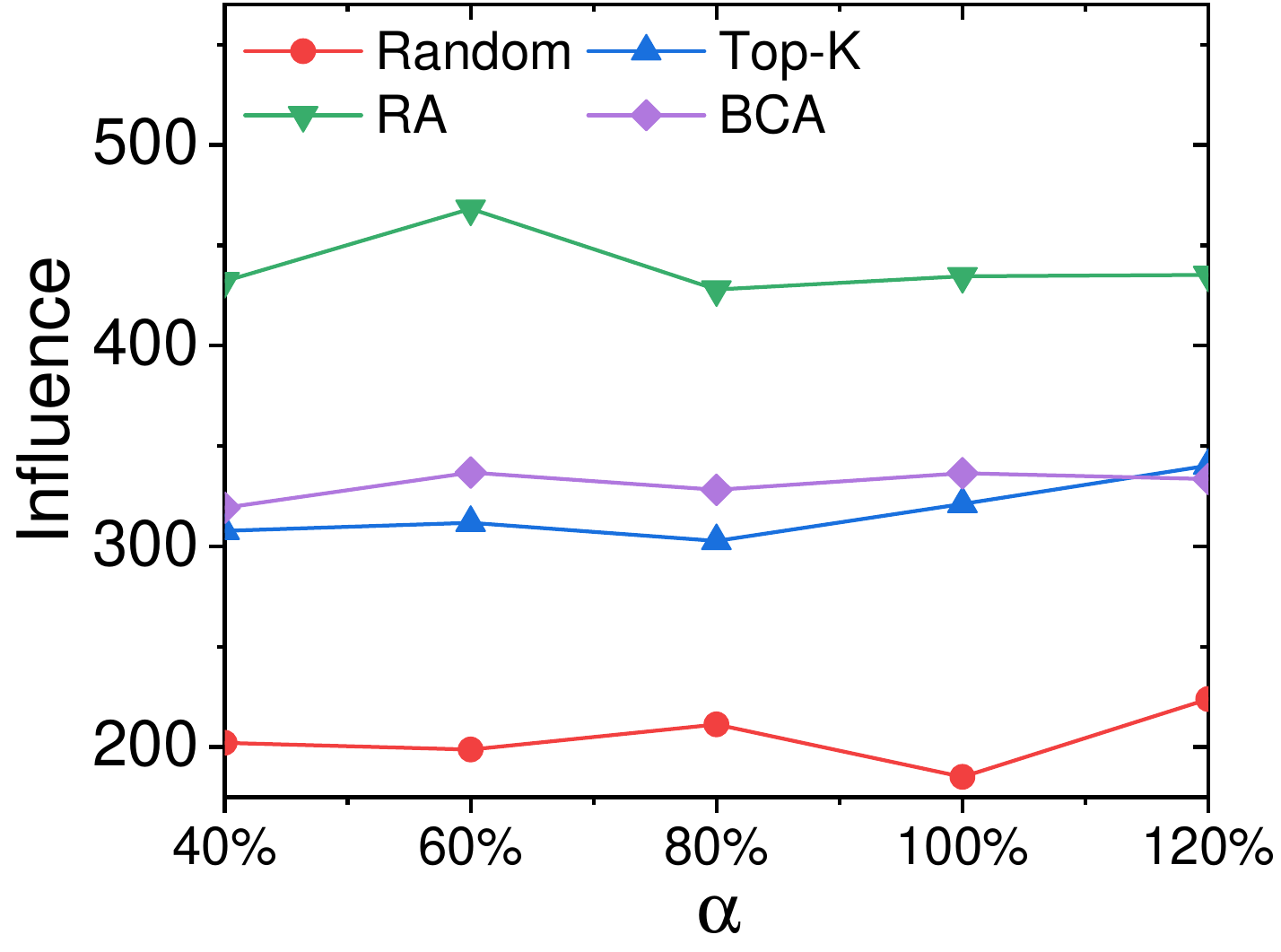} &
        \includegraphics[width=0.23\linewidth]{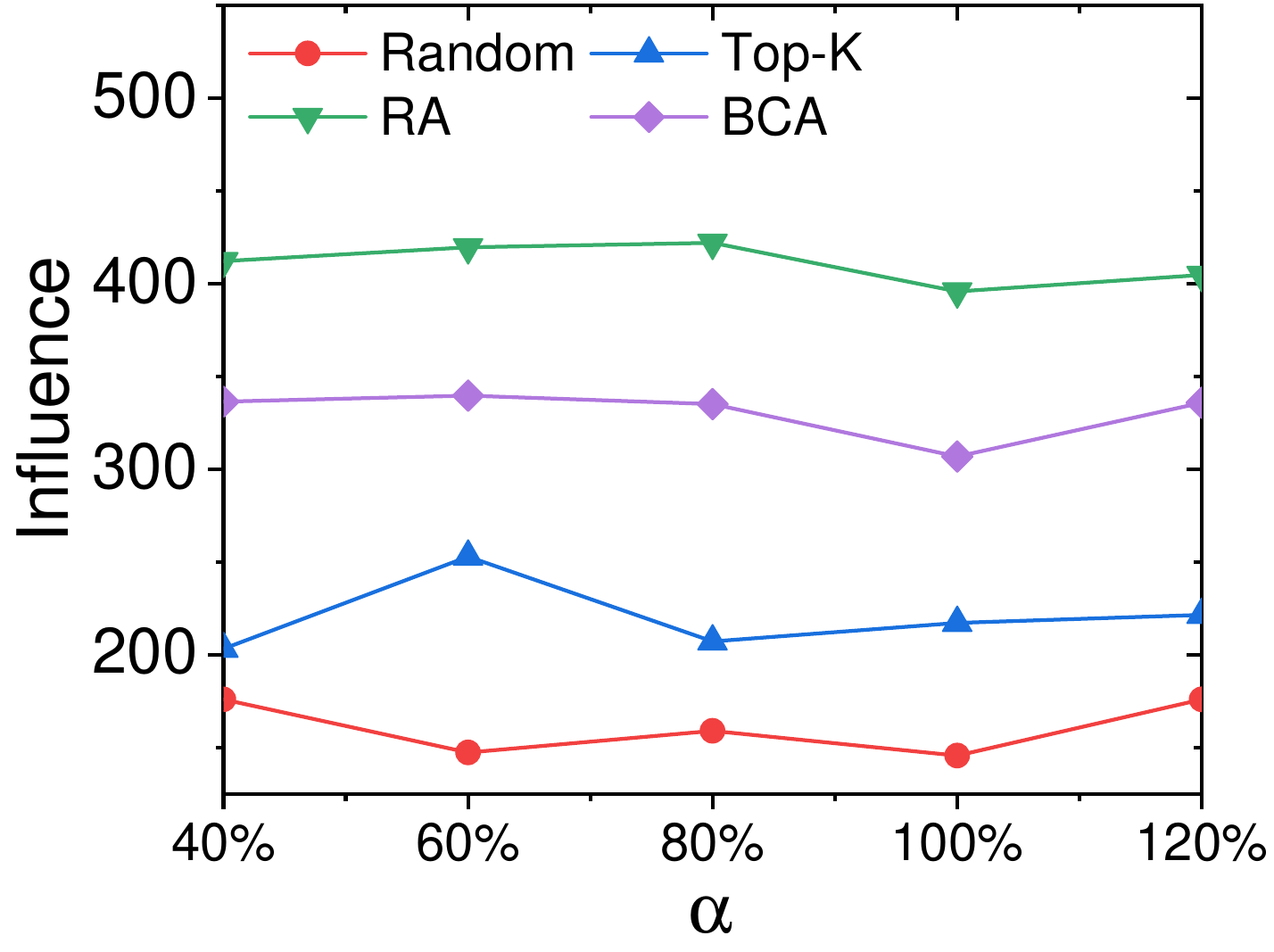} &
        \includegraphics[width=0.23\linewidth]{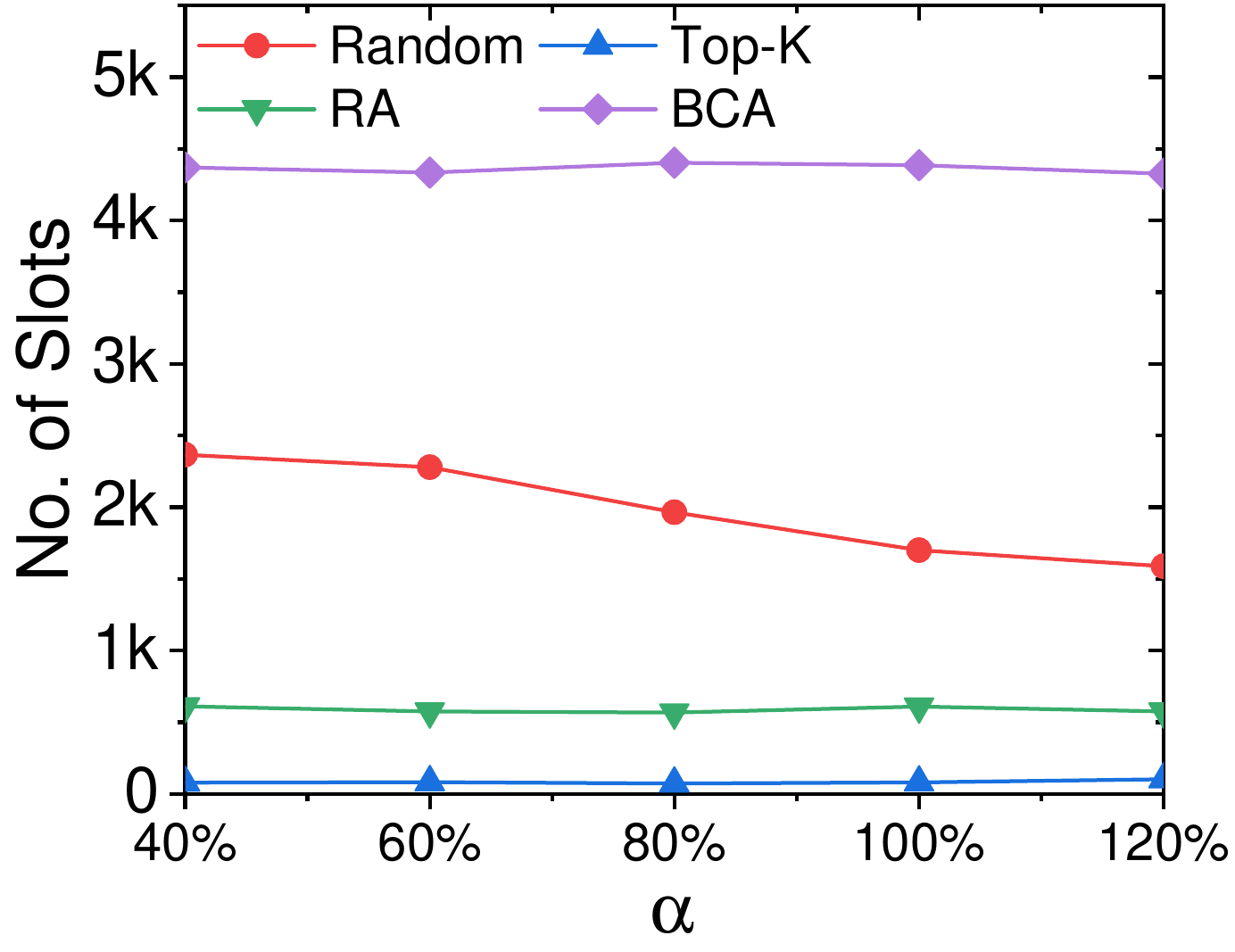} &
        \includegraphics[width=0.23\linewidth]{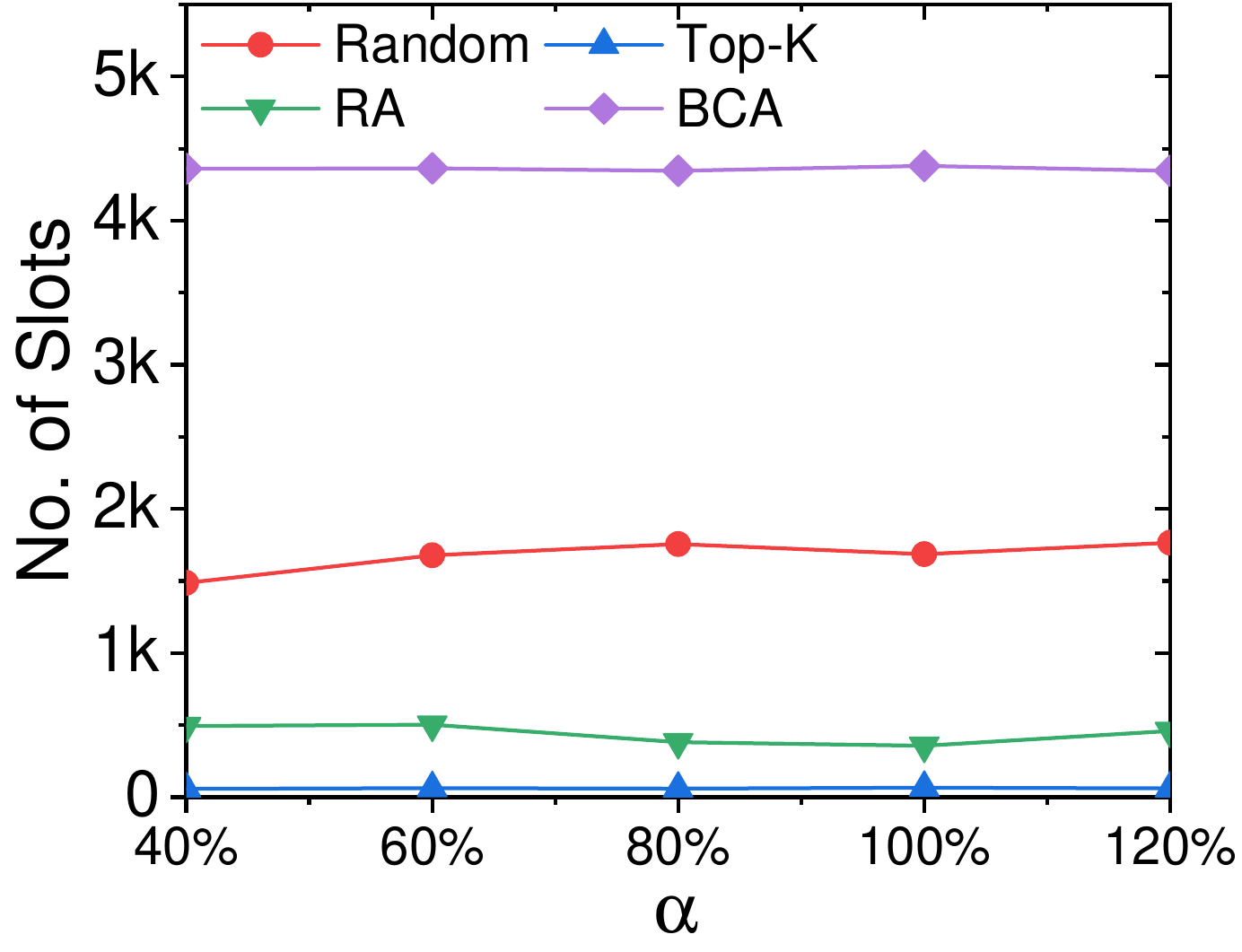} \\
        {\scriptsize (a) $\beta = 1\%,~ |\mathcal{P}| = 100$} &
        {\scriptsize (b) $\beta = 2\%,~ |\mathcal{P}| = 50$} &
        {\scriptsize (c) $\beta = 1\%,~ |\mathcal{P}| = 100$} &
        {\scriptsize (d) $\beta = 2\%,~ |\mathcal{P}| = 50$} \\
        \includegraphics[width=0.23\linewidth]{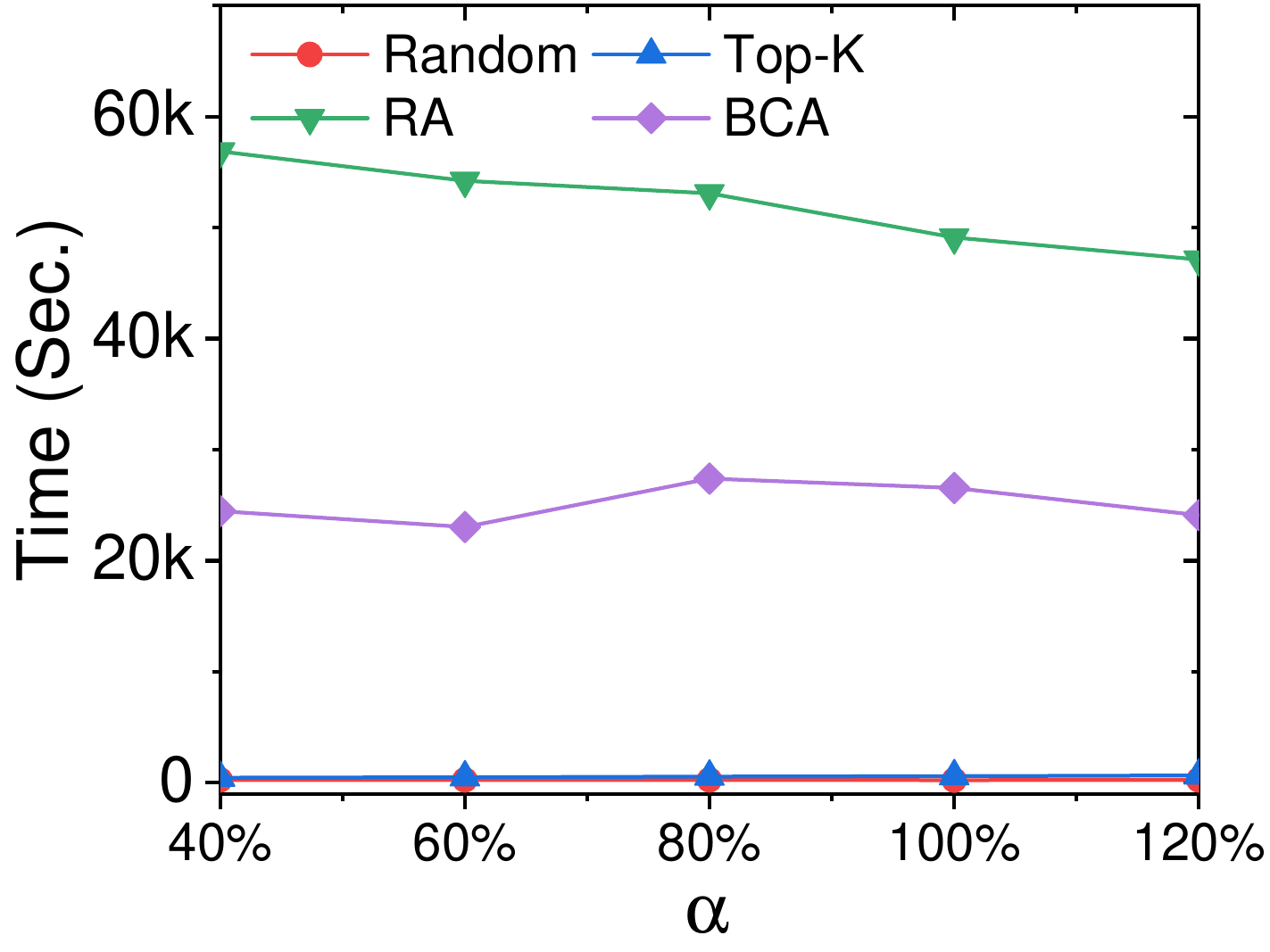} &
        \includegraphics[width=0.23\linewidth]{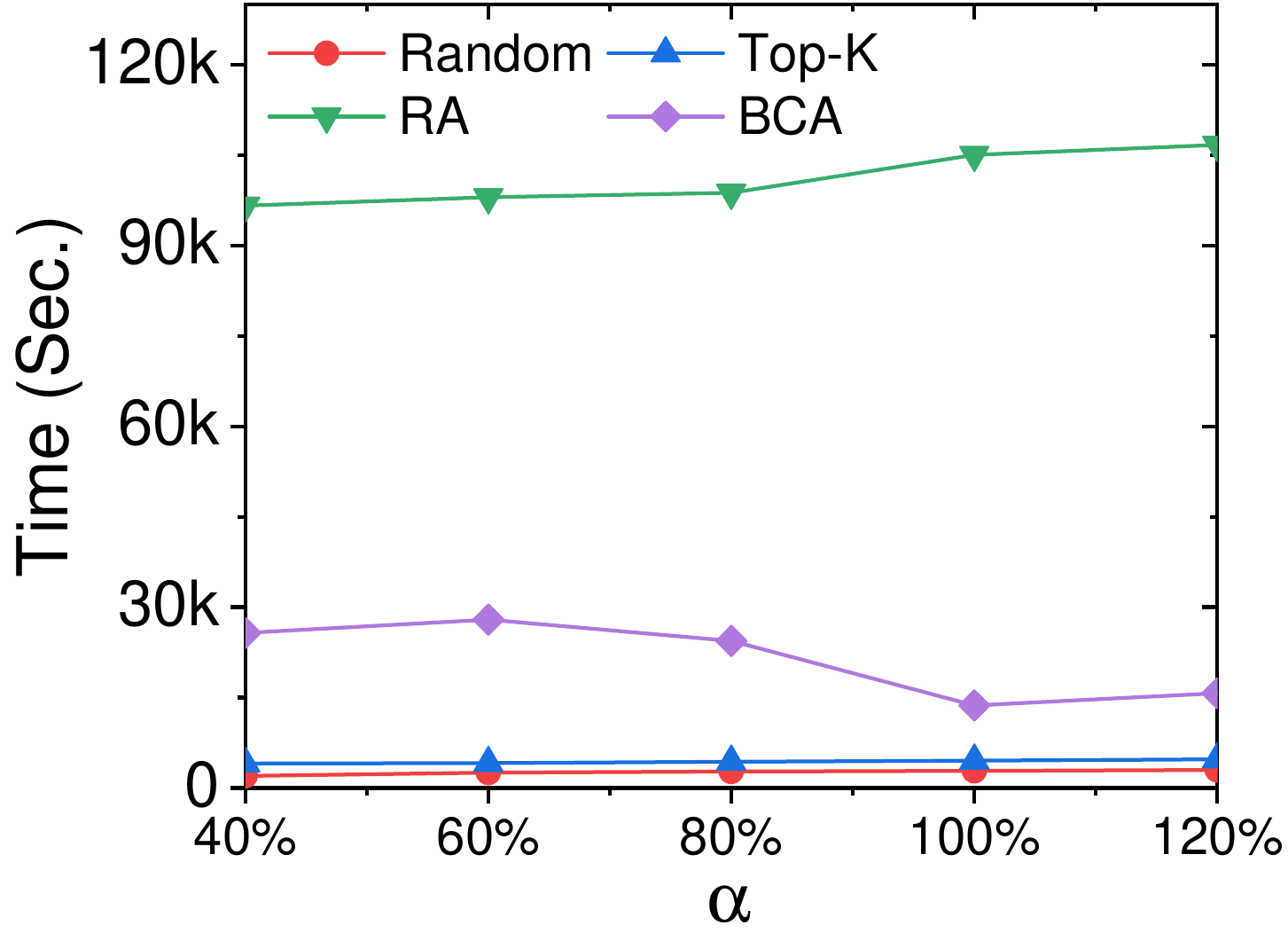} &
        \includegraphics[width=0.23\linewidth]{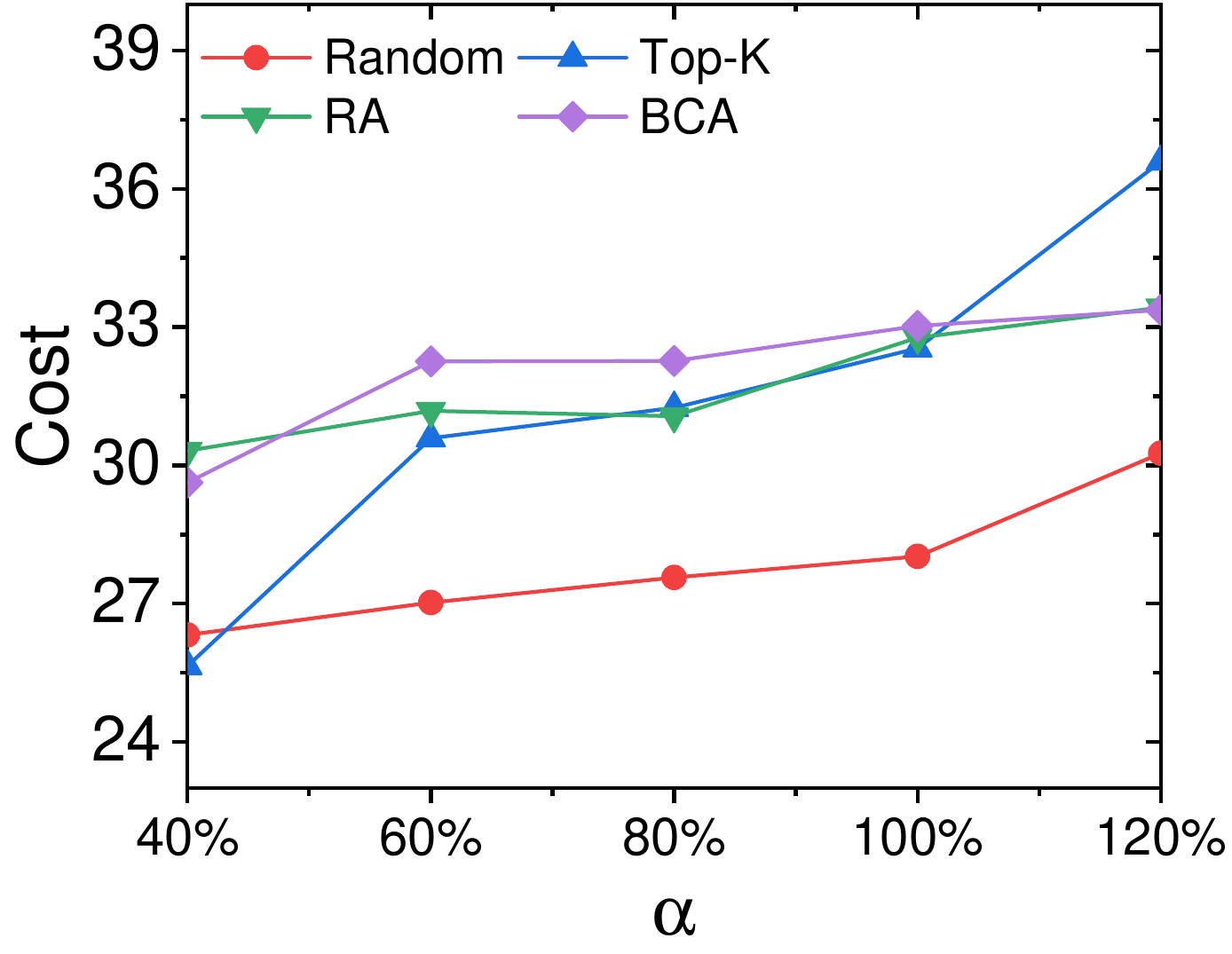} &
        \includegraphics[width=0.23\linewidth]{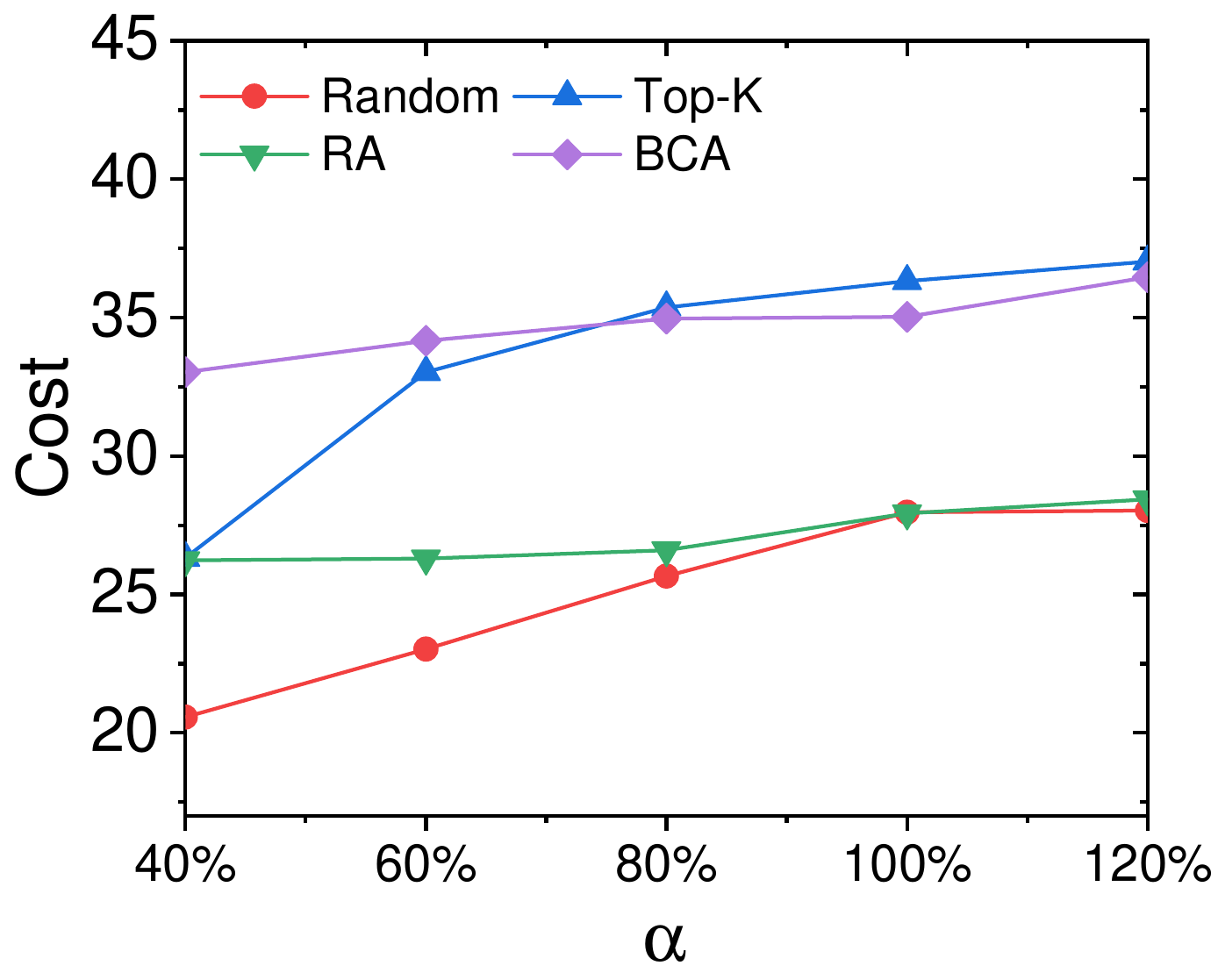} \\
        {\scriptsize (e) $\beta = 1\%,~ |\mathcal{P}| = 100$} &
        {\scriptsize (f) $\beta = 2\%,~ |\mathcal{P}| = 50$} &
        {\scriptsize (g) $\beta = 1\%,~ |\mathcal{P}| = 100$} &
        {\scriptsize (h) $\beta = 2\%,~ |\mathcal{P}| = 50$} \\
    \end{tabular}
    \caption{Plots for NYC dataset: $(a,b)$ $\alpha$ vs. Influence, $(c,d)$ $\alpha$ vs. No. of Slots, $(e,f)$ $\alpha$ vs. Time, and $(g,h)$ $\alpha$ vs. Cost}
    \label{Fig:NYC_Combined}
\end{figure*}

\paragraph{\textbf{Varying $\alpha$ and $\beta$ Vs. No. of Slots.}}
With an increase in $\alpha$ and $\beta$, the number of slots allocated for each product increases due to an increase in the demand for influence, as shown in Figure \ref{Fig:NYC_Combined}(c,d) and Figure \ref{Fig:LA_Combined} (c,d). Among the baseline approaches, the `Top-k' allocates fewer slots than the `Random' approach. As in the `Top-k', influential slots are stored in descending order and allocated to each product. The `BCA' approach allocates more slots than the `RA' and baseline methods, as it generates a single set of slots to satisfy all the product's demand. However, `RA' allocates a reasonable number of slots to satisfy the influence demand for all the advertisers.
\paragraph{\textbf{Efficiency Test.}}
From Figure \ref{Fig:NYC_Combined}(e,f) and Figure \ref{Fig:LA_Combined}(e,f), we have three main observations. First, the running time of `RA' is longer than that of the `BCA' and baseline methods because `RA' considers a large number of permutations of slots and products. Second, with increasing values of $\alpha$ and $\beta$, the computational time for all methods increases. Among the baseline approaches, `Top-k' takes more time compared to the `Random' approach. Third, the `BCA' approach takes more time compared to the baselines; however, it is less than the `RA' approach. With an increase of $\alpha$ and $\beta$, the computation time change is less because the initially selected slots using continuous greedy are significant, such that they satisfy most of the product's influence demand. So, further allocation using greedy based on marginal gain requires less time to compute.

\paragraph{\textbf{Product Vs. Budget.}}
From Figure \ref{Fig:NYC_Combined}(g,h) and Figure \ref{Fig:LA_Combined}(g,h), we have two observations. First, with increasing demand supply ratio $(\alpha)$ and average individual demand $(\beta)$, the influence of demand for each product increases, and a larger number of slots are allocated. Hence, the allocation cost of products increases. Second, in terms of cost, `RA' outperforms `Random', `Top-k', and `BCA' methods while providing maximum influence. `RA' uses the minimum cost to satisfy the influence demand for all the products.  

\begin{figure*}[ht]
    \centering
    \setlength{\tabcolsep}{2pt} 
    \renewcommand{\arraystretch}{0.9} 
    \begin{tabular}{cccc}
        \includegraphics[width=0.23\linewidth]{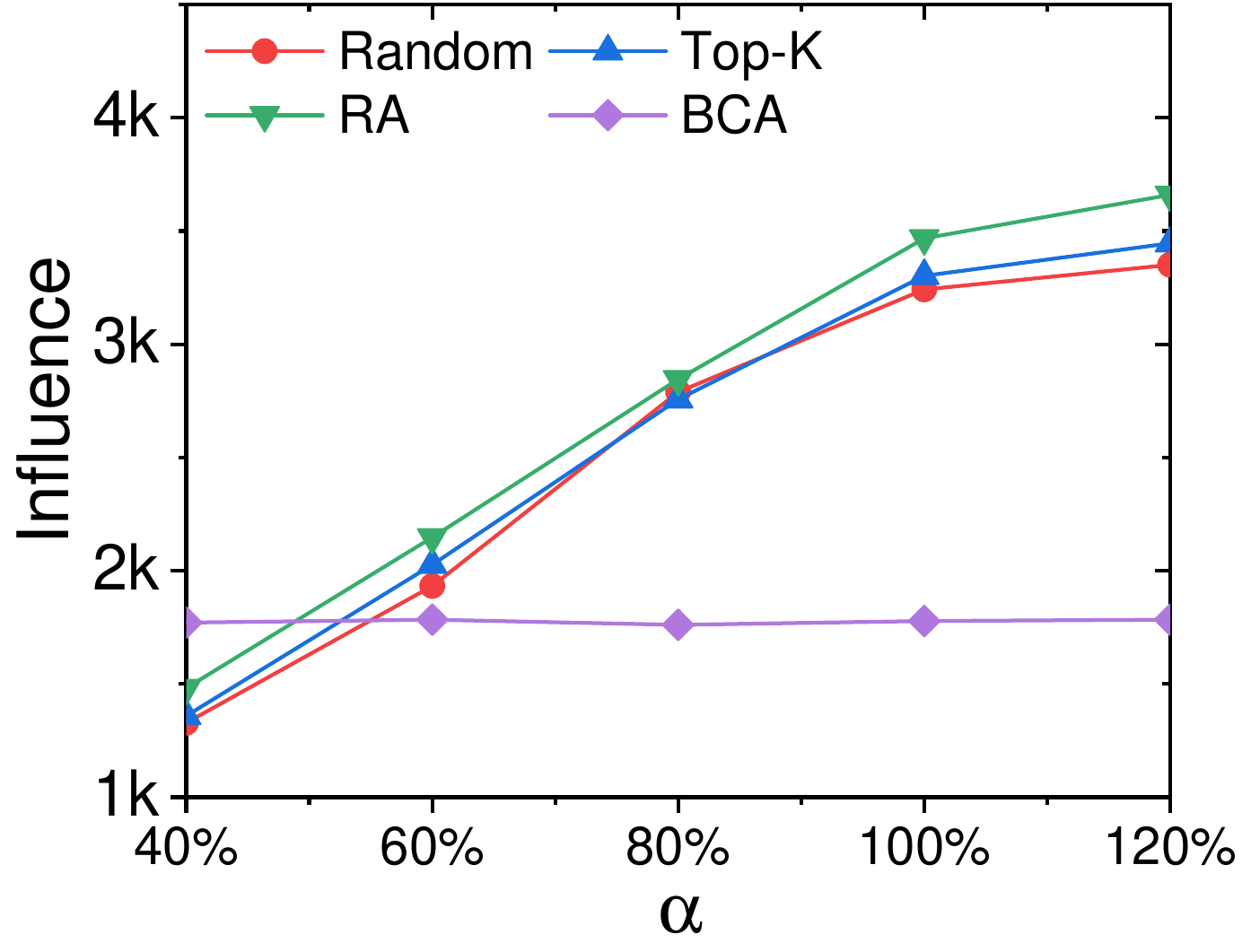} &
        \includegraphics[width=0.23\linewidth]{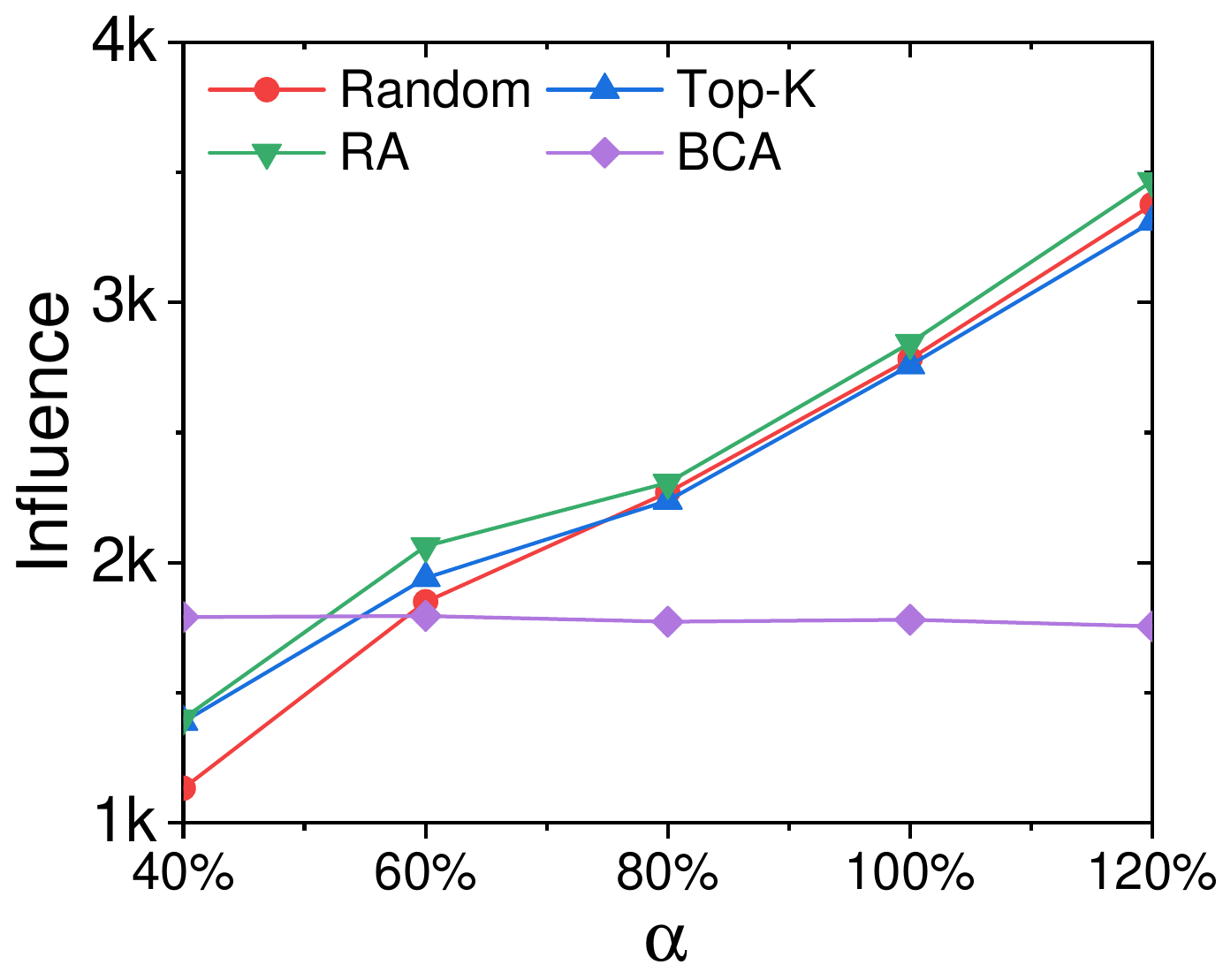} &
        \includegraphics[width=0.23\linewidth]{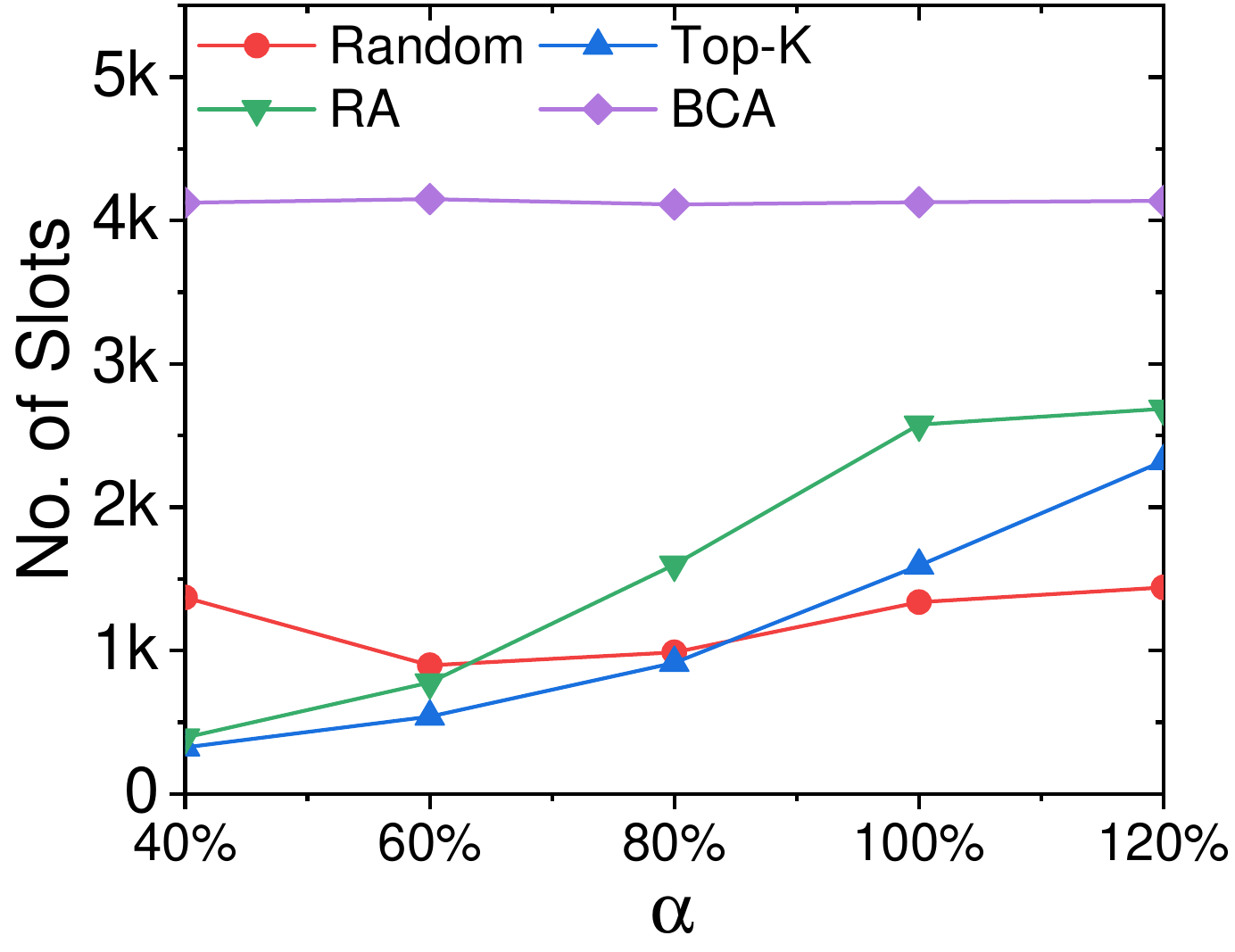} &
        \includegraphics[width=0.23\linewidth]{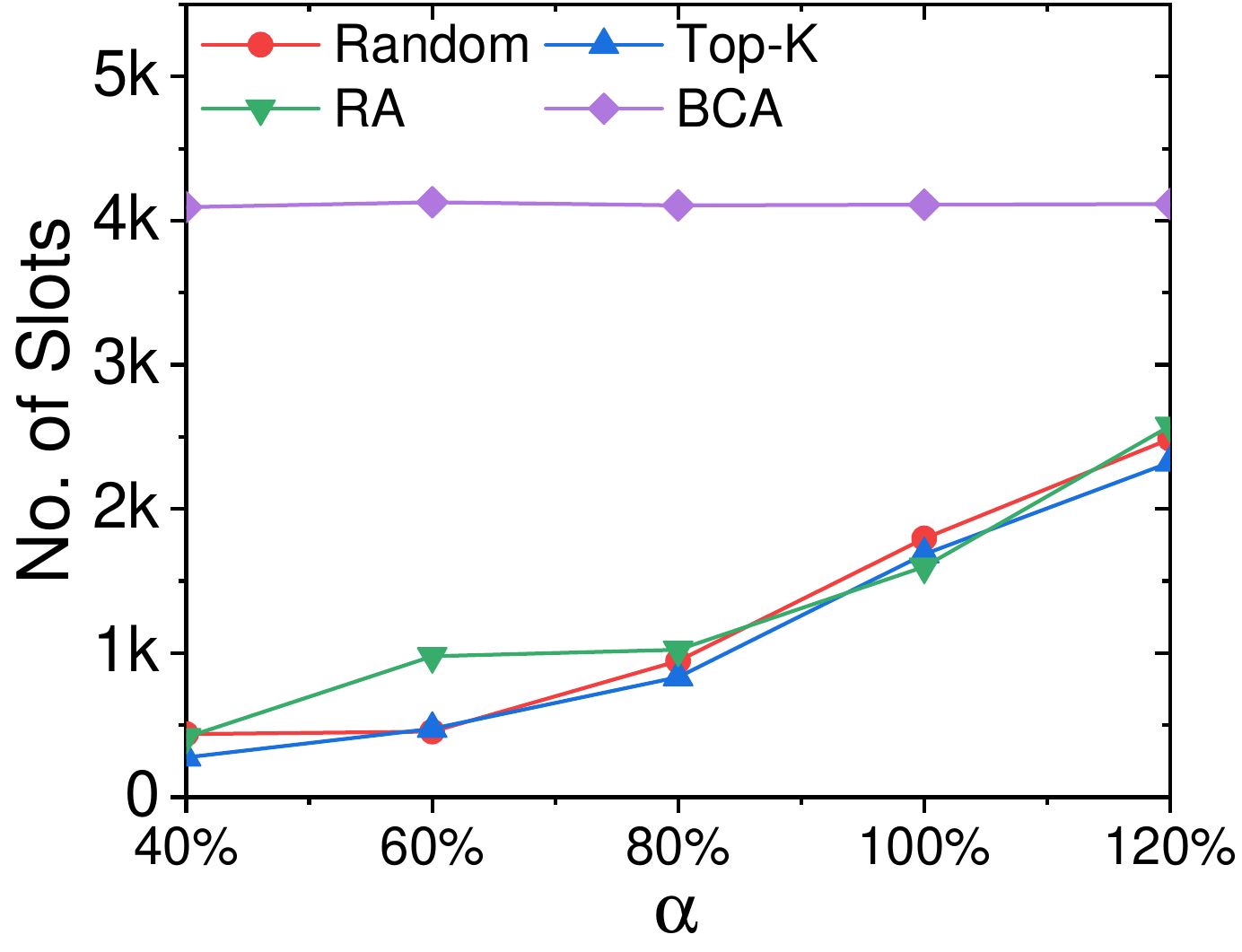} \\
        {\scriptsize (a) $\beta = 1\%,~ |\mathcal{P}| = 100$} &
        {\scriptsize (b) $\beta = 2\%,~ |\mathcal{P}| = 50$} &
        {\scriptsize (c) $\beta = 1\%,~ |\mathcal{P}| = 100$} &
        {\scriptsize (d) $\beta = 2\%,~ |\mathcal{P}| = 50$} \\
        \includegraphics[width=0.23\linewidth]{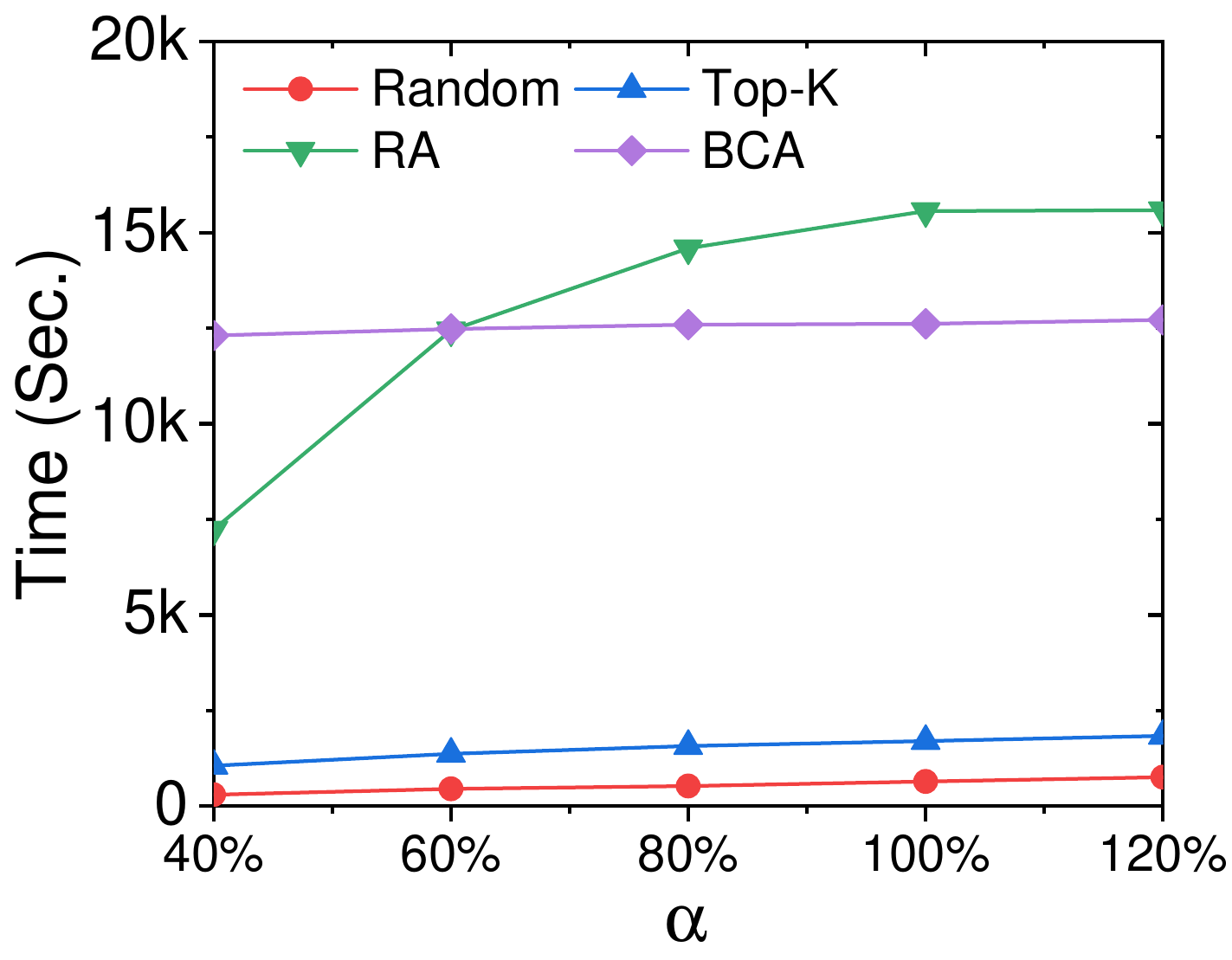} &
        \includegraphics[width=0.23\linewidth]{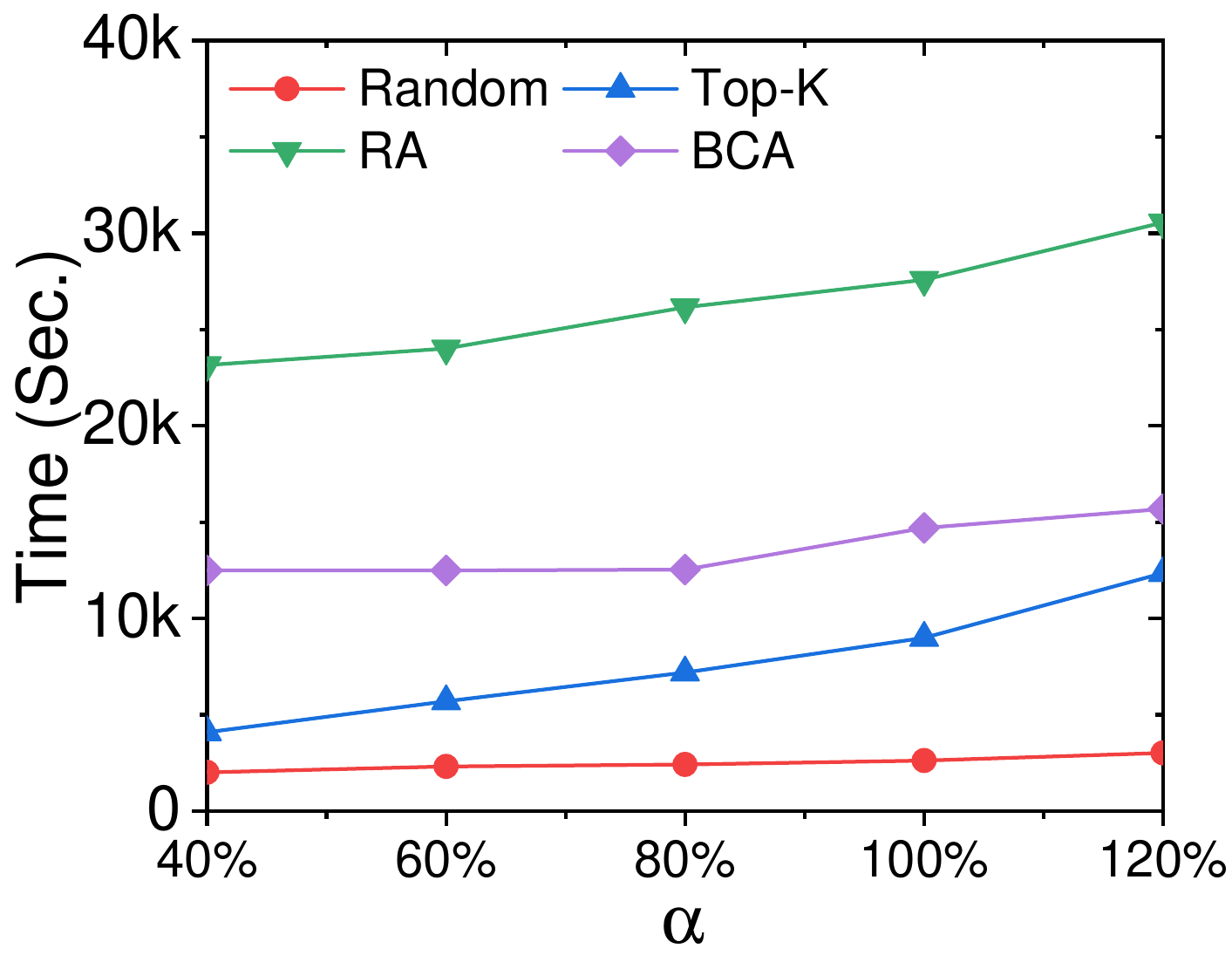} &
        \includegraphics[width=0.23\linewidth]{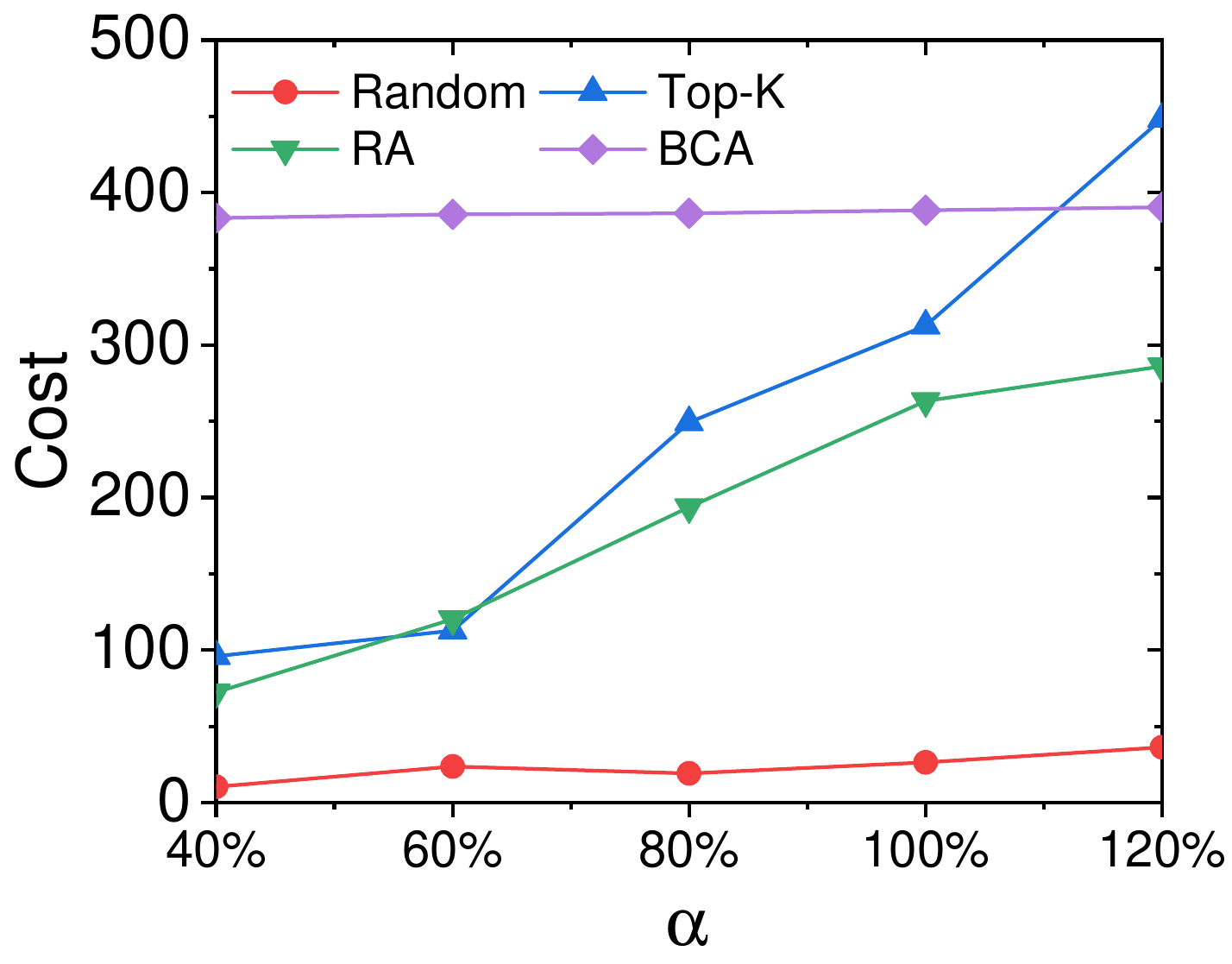} &
        \includegraphics[width=0.23\linewidth]{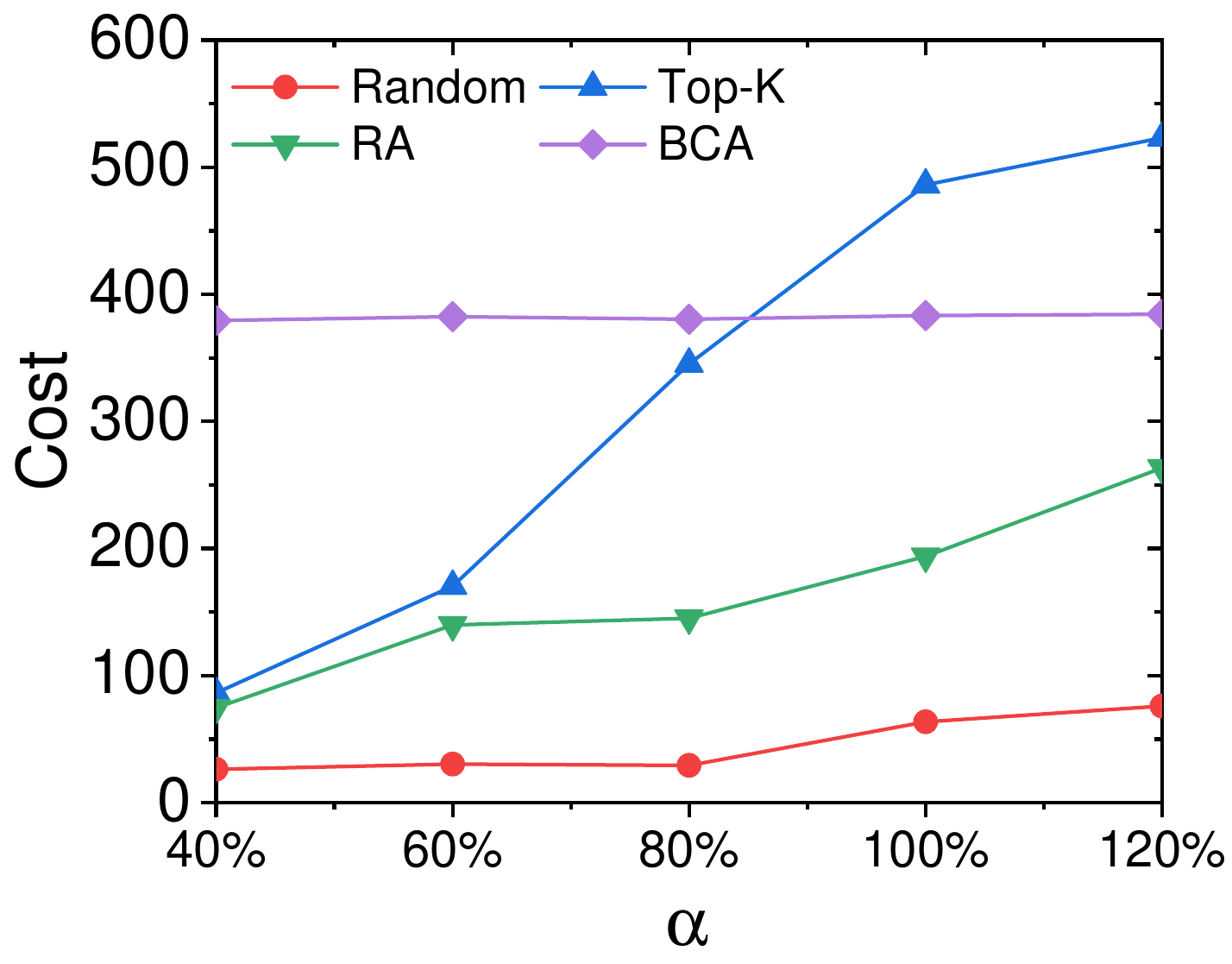} \\
        {\scriptsize (e) $\beta = 1\%,~ |\mathcal{P}| = 100$} &
        {\scriptsize (f) $\beta = 2\%,~ |\mathcal{P}| = 50$} &
        {\scriptsize (g) $\beta = 1\%,~ |\mathcal{P}| = 100$} &
        {\scriptsize (h) $\beta = 2\%,~ |\mathcal{P}| = 50$} \\
    \end{tabular}
    \caption{Plots for LA dataset: $(a,b)$ $\alpha$ vs. Influence, $(c,d)$ $\alpha$ vs. No. of Slots, $(e,f)$ $\alpha$ vs. Time, and $(g,h)$ $\alpha$ vs. Cost}
    \label{Fig:LA_Combined}
\end{figure*}

\paragraph{\textbf{Varying $\alpha$, $\beta$ Vs. Products.}}
We have experimented with different values of $\alpha$, $\beta$, and the number of products to show that different cases occur due to the varying demand of the products. Although we have reported the experimental results of two settings: $\alpha = 40\%, 60\%, 80\%, 100\%, 120\%$ and $\beta = 1\%, 2\%$, and the number of products $(|\ell|)$ is $100$ and $50$. However, the other cases are (1) when varying $\alpha$, $\beta = 5\%$, and  $|\ell| = 20$ (2) when varying $\alpha$, $\beta = 10\%$, and  $|\ell| = 10$ (3) when varying $\alpha$, $\beta = 20\%$, and  $|\ell| = 5$. In all these three cases, some of the product's influence demand is unsatisfied, and the `RA' methods will not provide any feasible solution due to a higher demand than supply. However, the `BCA', `Top-k', and `Random' provide solutions where some products are unsatisfied, and among them, `BCA' outperforms baseline methods.

\paragraph{\textbf{Additional Discussion.}}
The additional parameters used in our experiments are the distance parameter $\lambda$ and the approximation parameter $\epsilon$. The $\lambda$ determines the distance from which a billboard slot can influence a number of trajectories. The $\epsilon$ controls the accuracy of approximation. The smaller $\epsilon \in (0,1)$ provides a better approximation, but increases runtime. In our experiments, we set $\epsilon = 0.1$ and $\lambda = 100m$ as the default setting. Due to space limitations, we have only reported the results for the default settings. 

\section{Concluding Remarks} \label{Sec:CFD}
In this paper, we have studied the Multi-Product Influence Maximization Problem, where the goal is to maximize the influence for multiple products. As different trajectory user has affinity towards different products, it is a challenging problem. In this paper, we study this problem in two variants. In the first case, we select a set of $k$ slots for maximizing the aggregated influence, whereas in the second case, we are interested in selecting a non-intersecting set of slots depending upon the given budget. We modeled this problem as a submodular multi-cover problem and its generalized version. We have adopted the bi-criteria approximation algorithm for solving the first variant, and for the second variant, we have proposed a sampling-based randomized algorithm. The experimental evaluation with real-world datasets has demonstrated the superior performance of the proposed solution approaches.  
\bibliographystyle{ACM-Reference-Format}
\bibliography{sample-bibliography} 
\end{document}